\documentclass[11pt]{article} 

\usepackage[margin=1in]{geometry} 
\usepackage[utf8]{inputenc}
\usepackage[T1]{fontenc}
\usepackage[sc]{mathpazo} 

\usepackage{authblk}
\usepackage{comment,xspace}
\usepackage[normalem]{ulem} 

\usepackage{booktabs,multicol,multirow}
\usepackage{subcaption}

\usepackage{enumitem} 
\usepackage[dvipsnames]{xcolor} 
\definecolor{ptblue}{RGB}{15,76,129} 
\definecolor{ptemerald}{HTML}{009473} 
\definecolor{ptgray}{HTML}{939597} 

\usepackage{tikz} 
\usetikzlibrary{decorations.pathreplacing}

\usepackage{amsmath,amsfonts,amssymb,amsthm,bbm,mathtools}

\let\OLDland\land
\renewcommand{\land}{\:\OLDland\:}
\let\OLDlor\lor
\renewcommand{\lor}{\:\OLDlor\:}
\let\OLDforall\forall
\renewcommand{\forall}{\;\OLDforall\:}
\let\OLDexists\exists
\renewcommand{\exists}{\;\OLDexists\,}

\DeclareMathOperator*{\argmin}{arg\,min}
\newcommand{\indicator}[1]{\mathbbm{1}_{\left\{#1\right\}}\xspace}

\usepackage[ruled,linesnumbered,vlined]{algorithm2e}

\RequirePackage{boxedminipage}
\newcommand{\pbDef}[3]{%
	\noindent
	\begin{center}
		\begin{boxedminipage}{0.85\textwidth}
			 {#1}
			\smallskip\\
			\begin{tabular}{lp{0.9 \textwidth - \widthof{~~~Question:}}}
				Input: &  #2\\
				Question: &  #3
			\end{tabular}
		\end{boxedminipage}
	\end{center}
}

\usepackage[square,sort]{natbib} 

\usepackage{hyperref}
\hypersetup{
linktocpage,
colorlinks=true,
citecolor=ptemerald, 
urlcolor=ptblue, 
linkcolor=Plum, 
}
\usepackage{cleveref} 

\theoremstyle{plain}
\newtheorem{theorem}{Theorem}[section]

\newtheorem{lemma}[theorem]{Lemma}

\theoremstyle{definition}
\newtheorem{definition}[theorem]{Definition}
\newtheorem{example}[theorem]{Example}

\theoremstyle{remark}

\newcommand{\divGoods}{M}
\newcommand{\agents}{N}
\newcommand{\alloc}{A}
\newcommand{\MMS}{\text{MMS}}
\newcommand{\twoAgentOneHalfMMS}{$2$-\textsc{Agent}-$\frac{1}{2}$-\textsc{MMS}-\textsc{Alg}}

\title{\resizebox{\textwidth}{!}{Fair Allocation of Divisible Goods under Non-Linear Valuations}\thanks{A preliminary version of the paper appeared in \textit{Proceedings of the 24th International Conference on Autonomous Agents and Multiagent Systems (AAMAS)}~\citep{AzizHeLu25}.}}
\author[1]{Haris Aziz}
\author[2]{Zixu He}
\author[3]{Xinhang Lu}
\author[1]{Kaiyang Zhou}
\affil[1]{UNSW Sydney, \nolinkurl{haris.aziz@unsw.edu.au}, \nolinkurl{kaiyang.zhou@student.unsw.edu.au}}
\affil[2]{University of Massachusetts Amherst, \nolinkurl{taohe@umass.edu}}
\affil[3]{Kyushu University, \nolinkurl{xinhang.lu@inf.kyushu-u.ac.jp}}
\date{}

\begin{document}
\maketitle

\begin{abstract}
We study the problem of dividing homogeneous divisible goods among agents with non-linear valuations.
Specifically, the value that an agent gains from a given good depends only on the amount of the good they receive, and is not necessarily linear with respect to the amount.
For instance, under one-breakpoint piecewise-constant valuations, each agent specifies a threshold for each good such that this agent receives utility zero (resp., full utility of the good) when getting an amount below (resp., at least) the threshold.
Given non-linear valuations that are additive across the goods, we focus on designing fair allocation algorithms and consider two well-known fairness properties: the maximin share (MMS) guarantee and envy-freeness (EF).
For MMS, we devise an algorithm which always produces a $\frac{1}{2n-1}$-MMS allocation for $n$ agents with arbitrary non-decreasing valuations.
It is worth noting that this algorithmic result is almost tight as we give an impossibility of guaranteeing more than $1/n$ approximation to MMS, even when agents have one-breakpoint piecewise-constant valuations.
For $n \leq 3$ agents, we show the ratio~$1/n$ is tight.
Regarding envy-freeness, we show it is NP-hard to check the existence of an EF and Pareto optimal (PO) allocation for $n$ agents and at least three goods, even when agents have one-breakpoint piecewise-constant valuations.
We complement the hardness result by considering the case with a single divisible good, and devising a polynomial-time algorithm to check whether an EF and PO allocation exists or not for agents with piecewise-linear valuations.
\end{abstract}

\section{Introduction}

The allocation of scarce resources among multiple agents with different preferences is a fundamental issue that arises frequently in our society, for example, when dividing cloud computing resources such as processing time, memory, and communication bandwidth, or when handing out research grants.
We often want to ensure that the allocation is \emph{fair} to the agents, and possibly ideal in terms of other desiderata such as computation tractability, economic efficiency, etc.
When the resource to be allocated is divisible and \emph{heterogeneous}, the problem is commonly known as \emph{cake cutting}~\citep{Steinhaus49}, with the cake serving as a metaphor for heterogeneous divisible resource, and has been extensively studied by mathematicians, economists, political scientists, and more recently computer scientists~\citep{BramsTa96,LindnerRo24,Procaccia16,RobertsonWe98,Suksompong21}.

The rich literature of cake cutting provides several ways to capture fairness, with the two most prominent notions are \emph{proportionality} (a fair-share-based concept) and \emph{envy-freeness} (a comparison-based fairness concept).
An allocation is said to be proportional if each of the $n$ agents receives a utility of at least~$1/n$ of her total value for all resources~\citep{Steinhaus49}, and envy-free (EF) if every agent values her own bundle the most in the allocation~\citep{Foley67}.
Both notions can always be satisfied~\citep{Steinhaus49,Su99}, admit discrete and bounded protocols in the Robertson-Webb query model~\citep{AzizMa16-FOCS,RobertsonWe98}, and have been examined together with economic efficiency considerations such as being Pareto optimal (PO) or maximizing social welfare~\citep{BeiChHu12,BramsFeLa12,CohlerLaPa11,ReijniersePo98}.

As is common in the cake cutting literature, agents are assumed to have additive valuations across different pieces of the cake, yet the valuation function within each piece can be highly complicated.
Despite being a versatile model for fair division of divisible resources such as land, time, money, and computational resources, it fails to capture the natural scenario where agents care only about how much each divisible resource they receive rather than which part, and potentially do not have linear valuations in proportion to the amount they receive, as discussed in a few papers~\citep{BeiLiLu25,BuermannGeRa20,CaragiannisGkPs22,FeigeTe14}.
More specifically, consider the following real-world applications.
\begin{itemize}
\item \textbf{Computational resources.}
Given computational resources such as CPU time, memory and communication bandwidth to be divided between various users, each user requires at least a certain amount of resources to achieve meaningful performance in their computational tasks.
Some user may have CPU-intensive tasks and thus benefit more when receiving more CPU-time.
Another user may have a small file, so she gradually gets benefits in proportion to the size of her file being included in the memory but will stop getting more benefits once the file is fully included.

\item \textbf{Grant money.}
Agents' utilities are typically assumed to be proportional with respect to the amount of money they receive.
However, this may not be the case if, for example, a funding agency has specific guidelines on how much each budget category can be spent.
While some research groups may require equipment for experiments and be satisfied with receiving a large amount of funding in the ``Equipment'' category, other theory groups may find that most of the funding in this category is irrelevant.

\item \textbf{Space.}
Lastly, consider the scenario where a hall is shared among various community groups.
One community group needs at least half the hall to be able to organize a dinner event, and does not gets any additional benefit until it gets to book the whole hall in which case the full music system can also be played.
\end{itemize}

\emph{Given that agents have valuations depending on the amount of each homogeneous divisible resource they receive and the valuations are \emph{non-linear}, how shall we allocate the resources among agents in a fair manner?}

\subsection{Our Contributions}

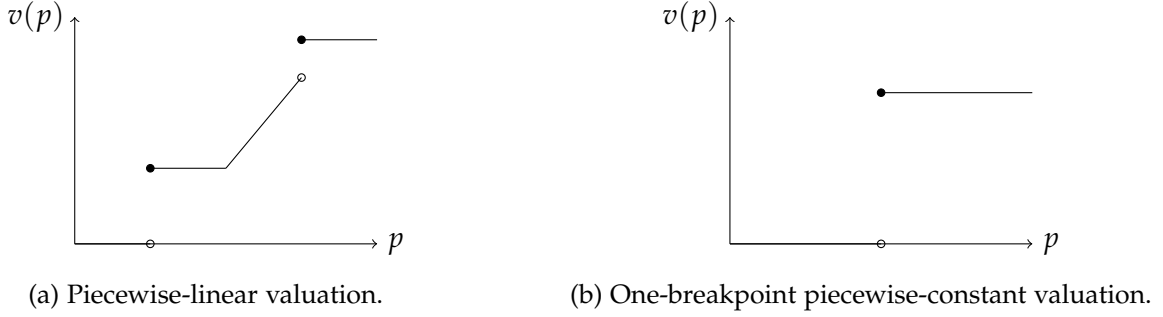
\begin{figure}[t]
\centering
\begin{subfigure}[b]{.45\linewidth}
\centering
\begin{tikzpicture}
\draw[->] (0,0) -- (4,0) node[right] {$p$};
\draw[->] (0,0) -- (0,3) node[left] {$v(p)$};
\draw[domain=0:1] plot (\x, {0});
\draw[domain=1:2] plot (\x, {1});
\draw[domain=2:3] plot (\x, {1.2*\x-1.4});
\draw[domain=3:4] plot (\x, {2.7});
\node[draw, circle, fill, inner sep=1pt] at (1,1) {};
\node[draw, circle, inner sep=1pt] at (1,0) {};
\node[draw, circle, inner sep=1pt] at (3,2.2) {};
\node[draw, circle, fill, inner sep=1pt] at (3,2.7) {};
\end{tikzpicture}
\caption{Piecewise-linear valuation.}
\label{fig:piecewise-linear}
\end{subfigure}
\hfill
\begin{subfigure}[b]{.5\linewidth}
\centering
\begin{tikzpicture}
\draw[->] (0,0) -- (4,0) node[right] {$p$};
\draw[->] (0,0) -- (0,3) node[left] {$v(p)$};
\draw[domain=0:2] plot (\x, {0});
\draw[domain=2:4] plot (\x, {2});
\node[draw, circle, fill, inner sep=1pt] at (2,2) {};
\node[draw, circle, inner sep=1pt] at (2,0) {};
\end{tikzpicture}
\caption{One-breakpoint piecewise-constant valuation.}
\label{fig:example-piecewise-constant-one-breakpoint}
\end{subfigure}
\caption{Examples of non-linear utility functions, where $v(p)$ specifies the value of a fraction~$p$ of the respective good under~$v$.}
\end{figure}

In this paper, we study the fair allocation of divisible resources and consider a model deviating from the cake-cutting literature.
In more detail, there are $m$ homogeneous divisible goods to be allocated among $n$ agents.
Each agent has a \emph{non-decreasing} utility function, which specifies the value this agent receives from a given good: the value depends only on the amount of the good she receives and is not necessarily linear with respect to the amount.
Agents' utilities are assumed to be additive across different goods.
It is worth noting that our model captures the setting of indivisible-goods allocation with additive utilities, by letting each agent have any value only if the agent gets a good in its entirety.
The \emph{maximin share (MMS) guarantee} is a natural relaxation of proportionality when allocating indivisible goods~\citep{Budish11}, and has been extensively studied in various fair division settings~\citep[see, e.g.,][]{AmanatidisAzBi23}.

The majority of our results concern structured utility functions such as being \emph{piecewise-linear} (see \Cref{fig:piecewise-linear}) or \emph{piecewise-constant with one breakpoint} (see \Cref{fig:example-piecewise-constant-one-breakpoint}).
The central fairness concepts in this paper are \emph{MMS} and \emph{envy-freeness}.
In \Cref{sec:MMS}, we focus on MMS, which is satisfied if agents get bundles worth at least their own maximin share---the largest value an agent can guarantee for herself if she partitions the goods into $n$ parts and gets the worst part.

When allocating indivisible goods to agents with additive utilities, a \emph{constant} multiplicative approximation to MMS can always be achieved~\citep{KurokawaPrWa18}, with the state-of-the-art factor of $\frac{7}{9}$ due to~\citet{HuangZh25} and an upper bound of~$\frac{39}{40}$ due to~\citet{FeigeSaTa21}.
In our model, however, it is impossible to guarantee more than $1/n$ approximation to MMS, even for one-breakpoint piecewise-constant valuations.
We complement this negative result by devising an algorithm that always produces a $\frac{1}{2n-1}$-MMS allocation for agents with \emph{arbitrary non-decreasing} valuations, asymptotically matching the upper bound.
We then turn to special cases involving up to three agents with one-breakpoint piecewise-constant valuations, where we show the approximation ratio~$1/n$ to MMS is \emph{tight}.
Indeed, our motivating examples demonstrate that one-breakpoint piecewise-constant valuations naturally capture a wide range of real-world applications.
Moreover, a good few fair division applications including dividing resources between different faculties within an institution or assets between founding members of a company often involves a small number of participants, and quite a few prominent fair division works deal exclusively with up to four agents~\citep[e.g.,][]{AmanatidisChFe18,AzizMa16,BergerCoFe22,BramsFi00,BramsKiKl14,ChaudhuryGaMe24,HollenderRu25,KilgourVe18}.

In \Cref{sec:EF}, we focus on envy-freeness, which can be satisfied trivially and vacuously by dividing each good equally among the agents.
We thus also aim to achieve economic efficiency, which is not guaranteed by dividing each good equally.
We show it is NP-hard to check the existence of an EF and Pareto optimal (PO) allocation, even for \emph{three} goods and one-breakpoint piecewise-constant valuations.
For a single good and piecewise-linear valuations, we devise a polynomial-time algorithm that (i) finds an allocation being EF and PO among all EF allocations, and (ii) checks the existence of an EF and PO allocation.

\subsection{Related Work}

Besides cake cutting, the fair allocation of indivisible goods has received extensive attention~\citep{AmanatidisAzBi23,NguyenRo23,Suksompong25}.
More recently, the fair allocation of resources of mixed nature has also been explored~\citep{LiuLuSu24}.

Perhaps the works most closely related to ours are the papers by \citet{CaragiannisGkPs22} and \citet{BeiLiLu25}.
\citeauthor{CaragiannisGkPs22} studied the fair allocation of homogeneous divisible goods in which each agent's value depends only on the amount of each good they receive and valuations are additive across different goods.
They mainly focused on randomized algorithms that achieve ex-ante envy-freeness and ex-ante approximate-PO among all envy-free lotteries; note that each deterministic allocation in the support is only required to be feasible.
They also considered more general valuations, and thus adopted a query model to elicit agents' valuations and focused on query complexity.
In contrast, we take a deeper dive into more structured valuations and focus on deterministic allocations which has guaranteed fairness ex post.
To this end, we also examine MMS that has not been explored for this model.

\citet{BeiLiLu25} studied a fair division model with \emph{subjective divisibility}, in which each good is either completely indivisible to some agents or completely divisible to other agents (i.e., the agents' valuation functions are linear with respect to the fraction of each good).
They adapted a hierarchy of envy-freeness relaxations, including EF1M~\citep{CaragiannisKuMo19}, EFM~\citep{BeiLiLi21} and EFXM~\citep{NishimuraSu25}, ordered by increasing strength, and investigated its compatibility with economic efficiency concepts.
Regarding the MMS guarantee, they showed an impossibility of guaranteeing more than $\frac{2}{3}$-MMS, and this is \emph{tight} for $2$ and $3$ agents.\footnote{\citet{BeiDiLi26} later devised polynomial-time algorithms that compute $\frac{2}{3}$-MMS allocations for up to four agents.}
A $\frac{1}{2}$-MMS allocation always exists for any number of agents.
The MMS approximation ratio was later improved to~$\frac{5}{9}$ in the work of \citet{BeiDiLi26}.
Moreover, \citet{BeiDiLi26} showed that the ratio~$\frac{2}{3}$ is tight in the case where all agents value all items equally.
Our model generalizes theirs as we allow a broader class of valuations.
Nevertheless, we still manage to give (asymptotically) tight approximation to MMS in various cases.

There have been works examining the division of a \emph{single} homogeneous divisible good with non-linear valuations~\citep{BuermannGeRa20,FeigeTe14}.
\citeauthor{FeigeTe14} presented a randomized allocation mechanism that gives each agent at least $1/2$ of her \emph{fair share} (i.e., the expected utility that she would get if she could choose the allocation rule that maximizes her expected utility).
\citeauthor{BuermannGeRa20} assumed that the available amount of the single divisible good is given by a probability distribution, and studied the trade-off between social welfare and envy-freeness as well as showed computational intractability of optimizing ex-ante social welfare subject to ex-ante EF (where randomness comes from the amount of the good).

When allocating resources of different types (e.g., computing resources) to agents with heterogeneous demands for each resource, a typical and classic assumption is that agents demand the resources in fixed proportions, known as \emph{Leontief preferences}~\citep{GhodsiZaHi11,ParkesPrSh15} in the economics literature.
This assumption requires that the resources are divisible and agents receive utilities in proportion to the resources allocated to them.
\citet{ParkesPrSh15} made a more practical assumption by requiring a minimum, \emph{indivisible} bundle of resources for the agents to receive utilities, and they conceptualized it as a step function.
Their negative results (on the incompatibility between PO, strategyproofness and fairness) hold with a single resource, meaning that these continue to hold in our setting.
The notable difference between their and our settings is that we assume additive utilities between the resources, instead of having the fixed proportion relations between the resources.

When allocating indivisible items, the MMS guarantee has been studied for valuations more general than additive.
The state-of-the-art MMS approximation guarantee for separable piecewise-linear concave (SPLC), submodular, fractionally subadditive (XOS) and subadditive valuations is $\frac{1}{2}$~\citep{ChekuriKuKu24}, $\frac{10}{27}$~\citep{BenUziahuFe23}, $\frac{3}{13}$~\citep{AkramiMeSe23} and $\frac{1}{14 \log n}$~\citep{FeigeHu25}, respectively.
Complementing the positive results, \citet{GhodsiHaSe22} provided an upper bound of~$\frac{3}{4}$ for submodular valuations and~$\frac{1}{2}$ for XOS (hence, subadditive) valuations.

\section{Preliminaries}

For any positive integer~$t$, let $[t] \coloneqq \{1, 2, \dots, t\}$.
Our model includes a set of~$n$ agents $N = [n]$ to whom we allocate a set of~$m$ homogeneous divisible goods $\divGoods = \{g_1, g_2, \dots, g_m\}$.
We assume in our work that each agent's value derived from each good depends on the fraction of the good allocated to the agent and that this derived value is not necessarily in proportion to the fraction allocated.
More precisely, each agent~$i \in N$ has a utility function $v_i \colon \divGoods \times [0, 1] \to \mathbb{R}_{\geq 0}$ such that $v_i(g, p)$ specifies the value of a fraction~$p$ of item~$g \in \divGoods$ allocated to agent~$i$.
Throughout the paper, we assume normalization, meaning that $v_i(g, 0) = 0$ for all~$i \in N$ and~$g \in \divGoods$, and monotonicity (a.k.a., ``free disposal''), meaning that for all~$i \in \agents$ and~$g \in \divGoods$, $v_i(g, p) \leq v_i(g, p')$ if $p \leq p'$.
For simplicity, we will write $v_i(g) = v_i(g, 1)$.

A \emph{bundle} of goods~$\divGoods$ is represented by an $m$-dimensional vector $\mathbf{x} = (x_1, x_2, \dots, x_m)$, where each coordinate~$x_j \in [0, 1]$ denotes the fraction of good~$g_j$ in bundle~$\mathbf{x}$.
For ease of expression, we will use ``$x_j \cdot g_j$'' to denote the part of good~$g_j$ allocated to a bundle.
The bundle~$\mathbf{x}$ is said to be \emph{integral} if all goods in the bundle are whole goods (as opposed to fractional goods), i.e., $x_j \in \{0, 1\}$ for all~$j \in [m]$.
In what follows, we will make it explicitly clear if we refer to an integral bundle.

We assume that agents' valuations across different goods are \emph{additive}, that is, given a bundle~$\mathbf{x}$, for all~$i \in \agents$, $v_i(\mathbf{x}) \coloneqq \sum_{j \in [m]} v_i(g_j, x_j)$.
Denote by $\alloc = (A_{i, g})_{i \in \agents, g \in \divGoods}$ the \emph{allocation} of goods~$\divGoods$ among the agents, where $A_{i, g}$ specifies the fraction of item~$g$ allocated to agent~$i$, and $\sum_{i \in N} A_{i, g} = 1$ for each $g\in \divGoods$ is the feasibility constraint.
Let $\alloc_i = (A_{i, g})_{g \in \divGoods}$ denote the bundle of goods allocated to agent~$i$.
Each agent~$i$ derives a utility of $v_i(A_i) \coloneqq \sum_{g \in \divGoods} v_i(g, A_{i, g})$ under allocation~$\alloc$.

Under this setting, we consider a few interesting sub-classes of general utility functions for the agents, namely piecewise-linear, piecewise-constant, and one-breakpoint piecewise-constant functions.
Note that each of the latter is a proper sub-class of the former classes.
A utility function~$v$ supported on~$[0, 1]$ is said to be
\begin{itemize}
\item \emph{piecewise-linear} if, for any~$g \in \divGoods$, $v(g, \cdot)$ can be written as a collection of linear functions; see \Cref{fig:piecewise-linear} for an illustration.
More formally, for some positive integer~$d$:
\begin{equation*}
v(g, p)=
\begin{cases}
a_1 \cdot p + b_1 & \text{if}~p \in [c_0, c_1) \\
a_2 \cdot p + b_2 & \text{if}~p \in [c_1, c_2) \\
\vdots \\
a_d \cdot p + b_d & \text{if}~p \in [c_{d-1}, c_d],
\end{cases}
\end{equation*}
where, $c_0 = 0$, $c_d = 1$, and for each~$j \in [d]$, $a_j, b_j \in \mathbb{R}$ are the coefficients of the corresponding linear function in segment $[c_{j-1}, c_j]$.
We say $c_1, c_2, \dots, c_{d-1}$ are the \emph{breakpoints} of the utility function.
Due to the normalization assumption, we have $b_1 = 0$.
Moreover, since all utility functions are assumed to be non-decreasing, we must have $a_j \geq 0$ for each $j \in [d]$ and $a_j \cdot c_j + b_j \leq a_{j+1} \cdot c_j + b_{j+1}$ for each~$j \in [d-1]$.

\item \emph{piecewise-constant} if, for any~$g \in \divGoods$, $v$ is piecewise-linear, and moreover, $a_j = 0$ for all~$j \in [d]$.

\item \emph{one-breakpoint piecewise-constant} if, for any~$g \in \divGoods$, $v$ is piecewise-constant, and furthermore, $d = 2$; see, e.g., \Cref{fig:example-piecewise-constant-one-breakpoint}.
\end{itemize}

We will refer to one-breakpoint piecewise-constant valuations frequently later.
To avoid ambiguity, we give an alternative form (that we will use more extensively) in the following: for each~$i \in \agents$ and~$g \in \divGoods$,
\[
v_i(g, p) =
\begin{cases}
v_i(g) & \text{if}~p \geq c_{i, g} \\
0      & \text{otherwise},         \\
\end{cases}
\]
where $c_{i, g} \in (0, 1]$ denotes the threshold of agent~$i$ for good~$g$ such that she starts getting positive utility.
The valuation can be succinctly represented by ``$v_i(g) \cdot \indicator{p \geq c_{i, g}}$'', where $\indicator{\cdot}$ is the indicator function that is~$1$ when agent~$i$ receives a fraction of good~$g$ at least the threshold~$c_{i, g}$ and~$0$ otherwise.
It is worth noting that when $c_{i, g} = 1$ for all~$i \in \agents$ and~$g \in \divGoods$, our setting becomes exactly the indivisible-goods allocation in which agents have additive utilities.

\section{Maximin Share Guarantee}
\label{sec:MMS}

In this section, we investigate the well-known share-based fairness notion---the \emph{maximin share (MMS) guarantee}, which was first proposed for indivisible goods~\citep{Budish11}, and has also been investigated in settings with a mix of both divisible and indivisible goods~\citep{BeiLiLu21,BeiLiLu25}.

\begin{definition}[$\alpha$-MMS]
Let~$\Pi_n(\divGoods)$ be the set of all $n$-partitions of~$\divGoods$.
The \emph{maximin share (MMS)} of~$i \in N$ is
\[
\MMS_i(n, \divGoods) \coloneqq \max_{(P_1, \dots, P_n) \in \Pi_n(\divGoods)} \min_{j \in [n]} v_i(P_j).
\]
Any partition for which the maximum is attained is called an \emph{MMS partition} of agent~$i$.

An allocation~$\alloc$ is said to satisfy the \emph{$\alpha$-MMS}, for some~$\alpha \in [0, 1]$, if for every agent~$i \in N$, $v_i(A_i) \geq \alpha \cdot \MMS_i(n, \divGoods)$.
\end{definition}

We demonstrate agents' MMS values and an approximate-MMS allocation below.

\begin{example}
    Consider an instance involving three agents $\{1, 2, 3\}$ and five goods $\{g_1, g_2, g_3, g_4, g_5\}$.
    Each agent has the one-breakpoint piecewise-constant valuation $0.4 \cdot \indicator{p \geq 0.3}$ for exactly two goods, and the valuation $0.2 \cdot \indicator{p \geq 0.8}$ for the remaining three goods.
    Specifically, agent~$1$ (resp., $2$ and $3$) regards goods $g_1, g_2$ (resp., $g_2, g_3$ and $g_3, g_4$) in manner as the first valuation.

    It can be verified that the MMS of each agent is~$1$.
    We illustrate an MMS partition of agent~$1$:
    \begin{itemize}
        \item each of the two goods with breakpoint at~$0.3$ (i.e., $g_1$ and $g_2$) is divided equally across three different bundles so that its value~$0.4$ is attained in all three bundles;
        \item the three remaining goods with breakpoint at~$0.8$ are allocated integrally to the three bundles respectively.
    \end{itemize}

    
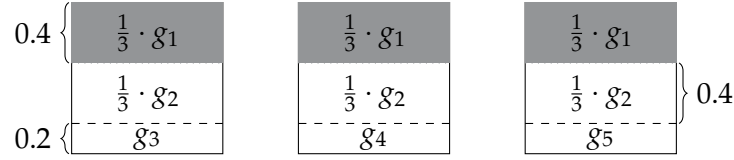
\begin{center}
        \begin{tikzpicture}[scale=2]
            \draw (0,0) rectangle (1,1);
            \node at (0.5,.1) {$g_3$};
            \draw[dashed] (0,.2) to (1,.2);
            \node at (.5,.4) {$\frac{1}{3} \cdot g_2$};
            \draw[dotted,color=ptgray] (0,.6) to (1,.6);
            \fill[ptgray] (0,.6) rectangle (1,1);
            \node at (.5,.8) {$\frac{1}{3} \cdot g_1$};

            \draw[decorate,decoration={brace,raise=1pt,amplitude=3pt}] (0,0) to (0,.2);
            \node[label=left:{$0.2$}] at (0,.1) {};
            \draw[decorate,decoration={brace,raise=1pt,amplitude=3pt}] (0,.6) to (0,1);
            \node[label=left:{$0.4$}] at (0,.8) {};

            \draw (1.5,0) rectangle (2.5,1);
            \node at (2,.1) {$g_4$};
            \draw[dashed] (1.5,.2) to (2.5,.2);
            \node at (2,.4) {$\frac{1}{3} \cdot g_2$};
            \draw[dotted,color=ptgray] (1.5,.6) to (2.5,.6);
            \fill[ptgray] (1.5,.6) rectangle (2.5,1);
            \node at (2,.8) {$\frac{1}{3} \cdot g_1$};

            \draw (3,0) rectangle (4,1);
            \node at (3.5,.1) {$g_5$};
            \draw[dashed] (3,.2) to (4,.2);
            \node at (3.5,.4) {$\frac{1}{3} \cdot g_2$};
            \draw[dotted,color=ptgray] (3,.6) to (4,.6);
            \fill[ptgray] (3,.6) rectangle (4,1);
            \node at (3.5,.8) {$\frac{1}{3} \cdot g_1$};
            \draw[decorate,decoration={brace,mirror,raise=1pt,amplitude=3pt}] (4,.2) to (4,0.6);
            \node[label=right:{$0.4$}] at (4,.4) {};
        \end{tikzpicture}
    \end{center}

    A similar MMS partition can be obtained for agents~$2$ and~$3$ by relabelling the goods according to their valuations.
    The following is a $0.8$-MMS allocation:
    \begin{itemize}
        \item Agent~$1$ receives $\{g_1, 0.5 \cdot g_2, g_5\}$ and gets utility~$1$.
        \item Agent~$2$ receives $\{0.5 \cdot g_2, 0.5 \cdot g_3\}$ and gets utility~$0.8$.
        \item Agent~$3$ receives $\{0.5 \cdot g_3, g_4\}$ and gets utility~$0.8$.
    \end{itemize}
\end{example}

\subsection{Any Number of Agents}

We start with an impossibility result.
Given a fair division instance, the \emph{MMS approximation guarantee} of the instance is the maximum~$\alpha$ such that the instance admits an $\alpha$-MMS allocation.
We show that the worst-case MMS approximation guarantee is at most~$1/n$.

\begin{theorem}
\label{thm:MMS:n-agent-impossibility}
For~$n$ agents with one-breakpoint piecewise-constant valuations, the worst-case MMS approximation guarantee is at most~$\frac{1}{n}$.
\end{theorem}

\begin{proof}
Consider the following instance with~$n$ goods.
\begin{center}
\begin{tabular}{@{}l|*{4}{l}@{}}
\toprule  & $v_i(g_1, \cdot)$                                           & $v_i(g_2, \cdot)$ & \ldots & $v_i(g_n, \cdot)$ \\
\midrule
Agent~$1$ & $\frac{1}{n} \cdot \indicator{A_{1, g_1} \geq \frac{1}{n}}$ & $\frac{1}{n} \cdot \indicator{A_{1, g_2} \geq \frac{1}{n}}$ & \ldots   & $\frac{1}{n} \cdot \indicator{A_{1, g_n} \geq \frac{1}{n}}$ \\
\midrule
Agent~$2$ & $\indicator{A_{2, g_1} = 1}$ & $\indicator{A_{2, g_2} = 1}$ & \ldots & $\indicator{A_{2, g_n} = 1}$ \\
\midrule
$\vdots$  & $\vdots$ & $\vdots$ & $\ddots$ & $\vdots$ \\
\midrule
Agent~$n$ & $\indicator{A_{n, g_1} = 1}$ & $\indicator{A_{n, g_2} = 1}$ & \ldots & $\indicator{A_{n, g_n} = 1}$ \\
\bottomrule
\end{tabular}
\end{center}

In words, for each good, agent~$1$ gets utility~$1/n$ when receiving at least $1/n$-th fraction of the good, and all other agents gets utility~$1$ when receiving the whole good.
The maximin share of each agent is exactly~$1$.
For agent~$1$, this can be seen by partitioning each good into~$n$ equal pieces and letting each bundle contain exactly one piece of each good.
For each of the remaining agents, having exactly one good in each bundle gives her MMS partition.

For any allocation in which all agents get positive utility, each agent of~$\{2, \dots, n\}$ needs to receive at least one good.
Then, agent~$1$ can only get a utility of at most~$1/n$.
The conclusion follows.
\end{proof}

In what follows, we provide an algorithm which always produce a $\frac{1}{2n - 1}$-MMS allocation for agents with non-decreasing valuations, asymptotically matching the upper bound.
The pseudocode can be found in \Cref{alg:1/(2n-1)-MMS}.
Note that our algorithmic result works for a much broader valuation class than that being used in the impossibility result.
We remark that one can adopt the \emph{value} and \emph{cut queries} from \citet{CaragiannisGkPs22} to access agents' non-decreasing valuations.
This section is most interested in existence results, and won't discuss query or computational complexities.

\begin{theorem}
    \label{thm:MMS:1/(2n-1)-MMS}
    For~$n$ agents with arbitrary non-decreasing valuation functions, \Cref{alg:1/(2n-1)-MMS} computes a $\frac{1}{2n - 1}$-MMS allocation.
\end{theorem}

\begin{algorithm}[t]
    \caption{$\frac{1}{2n - 1}$-MMS Allocation Algorithm}
    \label{alg:1/(2n-1)-MMS}
    \DontPrintSemicolon

    \KwIn{Agents~$N = [n]$ and goods~$\divGoods$.}
    \KwOut{A $\frac{1}{2n-1}$-MMS allocation~$\alloc$.}

    \ForEach{$i \in N$}{
        $A_i \gets \emptyset$\;
        Compute agent~$i$'s maximin share~$\MMS_i(n, \divGoods)$.\;
    }

    \BlankLine
    \tcp{Phase~I: Allocate large goods.}

    \While{$\exists g^* \in \divGoods$ s.t.\ $v_i(g^*) \geq \frac{\MMS_i}{2n - 1}$ for some~$i \in N$}{ \label{alg:1/(2n-1)-MMS:large-goods-condition}
        \ForEach{$i \in N$}{
            \eIf{$v_i(g^*) \geq \frac{\MMS_i}{2n - 1}$}{
                $x_{i, g^*} \gets \argmin_{p \in [0, 1]} v_i(g^*, p) \geq \frac{\MMS_i}{2n - 1}$\;
            }{
                $x_{i, g^*} \gets +\infty$\;
            }
        }
        Let~$S \subseteq N$ denote the subset of agents such that: (i) for all~$i \in S$ and for all~$j \in N \setminus S$, $x_{i, g^*} \leq x_{j, g^*}$, (ii) $\sum_{i \in S} x_{i, g^*} \leq 1$, and (iii) for all~$j \in N \setminus S$, $\sum_{i \in S \cup \{j\}} x_{i, g^*} > 1$.\; \label{alg:1/(2n-1)-MMS:identify-agents}
        \lForEach{$i \in S$}{
            $A_{i, g^*} \gets x_{i, g^*}$
        }
        $N \gets N \setminus S$, $\divGoods \gets \divGoods \setminus \{g^*\}$\; \label{alg:1/(2n-1)-MMS:large-goods-alloc}
    }

    \BlankLine
    \tcp{Phase~II: Bag-filling for small goods.}

    \While{$|N| > 1$}{ \label{alg:1/(2n-1)-MMS:small-goods-begin}
        Add one good at a time to an empty bundle~$B$ until $v_i(B) \in \left[ \frac{\MMS_i}{2n - 1}, \frac{2 \cdot \MMS_i}{2n - 1} \right]$ holds for some~$i \in N$.\; \label{alg:1/(2n-1)-MMS:bag-filling}
        $A_i \gets B$\;
        $N \gets N \setminus \{i\}$, $\divGoods \gets \divGoods \setminus B$\; \label{alg:1/(2n-1)-MMS:small-goods-alloc}
    }

    Give all remaining goods to the last agent.\; \label{alg:1/(2n-1)-MMS:last-agent}

    \Return{$(A_1, A_2, \dots, A_n)$}
\end{algorithm}

Given agents~$N$, goods~$\divGoods$ and agents' MMS values $(\MMS_i)_{i \in N}$, good~$g \in \divGoods$ is said to be a \emph{large good} if there exists some agent~$i \in N$ such that $v_i(g) \geq \frac{\MMS_i}{2n - 1}$; otherwise, we will say~$g$ is a \emph{small good}.
At a high level, \Cref{alg:1/(2n-1)-MMS} has two phases.
\Cref{alg:1/(2n-1)-MMS} starts by processing large goods in the first \verb|while|-loop (\crefrange{alg:1/(2n-1)-MMS:large-goods-condition}{alg:1/(2n-1)-MMS:large-goods-alloc}) by allocating the goods to agents.
After all large goods are allocated, the remaining goods are allocated to the remaining agents in \crefrange{alg:1/(2n-1)-MMS:small-goods-begin}{alg:1/(2n-1)-MMS:small-goods-alloc} via \emph{bag-filling}.
We remark that although the last agent gets all remaining goods, this agent may not necessarily receive a lot more value due to her non-linear utility over each good.
Note that in the following, $\MMS_i$ refers to $\MMS_i(n, \divGoods)$ (i.e., the maximin share computed in the original $n$-agent instance).
Sometimes, we will refer to maximin share in reduced instances---we will make it clear.

First, we investigate instances that have only small goods.

\begin{lemma}
    \label{lem:1/(2n-1)-MMS:small-goods}
    Given $\langle [n], \divGoods \rangle$ in which $v_i(g) < \frac{1}{2n - 1} \cdot \MMS_i$ for all~$i \in N$ and~$g \in \divGoods$, \Cref{alg:1/(2n-1)-MMS} computes a $\frac{1}{2n - 1}$-MMS allocation.
\end{lemma}

\begin{proof}
    Given the $n$-agent instance, the first \verb|while|-condition in \cref{alg:1/(2n-1)-MMS:large-goods-condition} is evaluated as false due to the assumption in the \namecref{lem:1/(2n-1)-MMS:small-goods} statement that for all agents~$i \in N$ and goods~$g \in \divGoods$, $v_i(g) < \frac{\MMS_i}{2n - 1}$.
    \Cref{alg:1/(2n-1)-MMS} thus executes \crefrange{alg:1/(2n-1)-MMS:small-goods-begin}{alg:1/(2n-1)-MMS:small-goods-alloc} to process this instance.
    \Cref{alg:1/(2n-1)-MMS:bag-filling} adds one good at a time to an empty bundle~$B$ until there exists some agent~$i \in N$ such that $v_i(B) \geq \frac{\MMS_i}{2n - 1}$.
    As we have assumed that agents' valuations across goods are additive, such an agent~$i$ always exists.
    We allocate bundle~$B$ to agent~$i$ and remove both the agent and her goods from further consideration.
    Clearly, agent~$i$ is satisfied with her bundle and gets a utility of at least~$\frac{1}{2n - 1} \cdot \MMS_i$.

    For all agents~$j \in N$, we have $v_j(B) \leq \frac{2 \cdot \MMS_j}{2n - 1}$, because $v_j(g) < \frac{\MMS_j}{2n - 1}$ for all~$g \in \divGoods$.
    It implies that in the reduced instance $\langle N \setminus \{i\}, \divGoods \setminus B \rangle$, for all agents~$j \in N \setminus \{i\}$,
    \[
        \MMS_j(n-1, \divGoods \setminus B) \geq \left( 1 - \frac{2}{2n - 1} \right) \cdot \MMS_j.
    \]
    Put differently, whenever we remove one agent and their goods from consideration in \Cref{alg:1/(2n-1)-MMS}, the remaining agents' MMS values in the reduced instance decrease by at most~$\frac{2}{2n - 1}$ of their MMS values computed in the original ($n$-agent) instance.

    By the design of the algorithm, it is clear that each removed agent in \crefrange{alg:1/(2n-1)-MMS:small-goods-begin}{alg:1/(2n-1)-MMS:small-goods-alloc} is satisfied with their received bundle and gets a utility of at least $\frac{1}{2n - 1}$ of their own maximin share (computed in the original $n$-agent instance).
    It remains to show when there is only one agent left, this last agent is satisfied with receiving all remaining goods.
    This can be seen from the fact that all remaining goods is worth at least $1 - \frac{2 \cdot (n - 1)}{2n - 1} = \frac{1}{2n - 1}$ of her maximin share.
\end{proof}

In what follows, we investigate instances having large goods.

\begin{lemma}
    \label{lem:1/(2n-1)-MMS:large-goods}
    Given any instance with agents~$N = [n]$, goods~$\divGoods$ and $v_i(g) \geq \frac{1}{2n - 1} \cdot \MMS_i$ for some~$i \in N$ and~$g \in \divGoods$, \crefrange{alg:1/(2n-1)-MMS:large-goods-condition}{alg:1/(2n-1)-MMS:large-goods-alloc} of \Cref{alg:1/(2n-1)-MMS} computes an allocation $(A_i)_{i \in N'}$ for agents~$N' \subseteq N$ and a reduced instance with agents~$N'' = N \setminus N'$ and goods~$\divGoods'' = \divGoods \setminus \bigcup_{i \in N'} A_i$ such that
    \begin{enumerate}[label=(\roman*)]
        \item for each~$i \in N'$, $v_i(A_i) \geq \frac{\MMS_i}{2n - 1}$;
        \item for each~$i \in N''$ and~$g \in \divGoods''$, $v_i(g) < \frac{\MMS_i}{2n - 1}$;
        \item for each~$i \in N''$, $\MMS_i(|N''|, \divGoods'') \geq |N''| \cdot \frac{2 \cdot \MMS_i}{2n - 1}$.
    \end{enumerate}
\end{lemma}

\begin{proof}
It is easy to see that the first two properties hold by the design of \Cref{alg:1/(2n-1)-MMS}.
Specifically, each agent~$i \in S$ removed in \cref{alg:1/(2n-1)-MMS:large-goods-alloc} gets a utility of at least~$\frac{\MMS_i}{2n - 1}$ by receiving a fraction of good.
Moreover, when \crefrange{alg:1/(2n-1)-MMS:large-goods-condition}{alg:1/(2n-1)-MMS:large-goods-alloc} of \Cref{alg:1/(2n-1)-MMS} terminate, clearly, each remaining agent~$i \in N''$ values each remaining good~$g \in \divGoods''$ less than~$\frac{\MMS_i}{2n - 1}$.
It remains to show the last property.

Fix any agent~$i \in N''$ and her MMS partition of the original goods~$\divGoods$.
For ease of expression, assume that the \verb|while|-loop in \crefrange{alg:1/(2n-1)-MMS:large-goods-condition}{alg:1/(2n-1)-MMS:large-goods-alloc} executes~$s$ iterations, and in each iteration~$j \in [s]$, good~$g_j^*$ is allocated to agents~$N_j$.
That is, $N' = \bigcup_{j \in [s]} N_j$.
By the design of the algorithm (see \cref{alg:1/(2n-1)-MMS:identify-agents}), in each iteration~$j \in [s]$, agent~$i$ needs at least a fraction $\frac{1}{|N_j|}$ of good~$g_j^*$ to reach value $\frac{\MMS_i}{2n - 1}$.
Put differently, for at least $n - |N_j|$ bundles in the MMS partition, goods~$g_j^*$ contributes a value of less than $\frac{\MMS_i}{2n - 1}$ to each of the bundle.
As a result, when \crefrange{alg:1/(2n-1)-MMS:large-goods-condition}{alg:1/(2n-1)-MMS:large-goods-alloc} terminate, there are at least $n - |N'| = |N''|$ bundles, each of which is worth at least
\[
\MMS_i - \sum_{j \in [s]} \frac{\MMS_i}{2n - 1} = \left( 1 - \frac{|N'|}{2n - 1} \right) \cdot \MMS_i = \frac{(n - 1) + (n - |N'|)}{2n - 1} \cdot \MMS_i \geq \frac{2 \cdot |N''|}{2n - 1} \cdot \MMS_i,
\]
where the last transition follows from the fact that the \verb|while|-loop in \crefrange{alg:1/(2n-1)-MMS:large-goods-condition}{alg:1/(2n-1)-MMS:large-goods-alloc} executes at least once, and thus $n - 1 \geq |N''|$.
It follows that $\MMS_i(|N''|, \divGoods'') \geq |N''| \cdot \frac{2 \cdot \MMS_i}{2n - 1}$.
\end{proof}

\begin{lemma}
    \label{lem:1/(2n-1)-MMS:small-goods-reduced}
    Given any (reduced) instance with agents~$N''$ and goods~$\divGoods''$, \crefrange{alg:1/(2n-1)-MMS:small-goods-begin}{alg:1/(2n-1)-MMS:small-goods-alloc} of \Cref{alg:1/(2n-1)-MMS} computes an $\frac{1}{2n - 1}$-$\MMS$ allocation for the original instance $\langle N, M \rangle$.
\end{lemma}

\begin{proof}
By \Cref{lem:1/(2n-1)-MMS:large-goods}, it follows that for each agent $i \in N''$:
\begin{enumerate}[label=(\roman*)]
\item $v_i(g) < \frac{\MMS_i}{2n - 1}$ for each $g \in \divGoods''$; and
\item $\MMS_i(|\agents''|, \divGoods'') \geq |N''| \cdot \frac{2 \cdot \MMS_i}{2n - 1}$.
\end{enumerate}
Combining the two results yields that, for each $g\in \divGoods''$,
\[
\frac{\MMS_i(|N''|,M'')}{2 |N''| - 1} > \frac{\MMS_i(|N''|,M'')}{2 |N''|} \geq \frac{\MMS_i}{2n - 1} \geq v_i(g),
\]
which implies that all goods in $\divGoods''$ are small under $\langle N'', \divGoods'' \rangle$.

However, by \Cref{lem:1/(2n-1)-MMS:small-goods}, \crefrange{alg:1/(2n-1)-MMS:small-goods-begin}{alg:1/(2n-1)-MMS:small-goods-alloc} guarantees a $\frac{1}{2|N''| - 1}$-MMS allocation for $\langle N'', \divGoods'' \rangle$. Since $\frac{\MMS_i(|N''|,M'')}{2 |N''| - 1} \geq \frac{\MMS_i}{2n - 1}$, this allocation is also a $\frac{1}{2n - 1}$-MMS allocation for the original instance $\langle N, \divGoods \rangle$.
\end{proof}

The correctness of \Cref{thm:MMS:1/(2n-1)-MMS} follows from \Cref{lem:1/(2n-1)-MMS:small-goods,lem:1/(2n-1)-MMS:large-goods,lem:1/(2n-1)-MMS:small-goods-reduced}.

\subsection{A Small Number of Agents}

We now consider cases involving a small number of agents who have one-breakpoint piecewise-constant valuations over the goods.

Our main result in this subsection is the following: For $n \leq 3$ agents with one-breakpoint piecewise-constant valuations, there always exists a $\frac{1}{n}$-MMS allocation.
It is worth noting that our algorithmic results match the upper bounds stated in \Cref{thm:MMS:n-agent-impossibility}.
For $n \leq 3$ agents, we give a tight approximation ratio for MMS.

Our algorithmic result holds trivially when $n = 1$.
In what follows, we will first present an algorithm that can always output a $\frac{1}{2}$-MMS allocation for two agents in \Cref{sec:MMS:piecewise-const:2-agent}.
Built up on this $2$-agent-$\frac{1}{2}$-MMS algorithm, we will then present an algorithm that always outputs a $\frac{1}{3}$-MMS allocation for three agents in \Cref{sec:MMS:piecewise-const:3-agent}.
To establish the algorithms, we present some auxiliary lemmas which will be proven useful.
These lemmas hold for any number of agents, and connect an original instance and its reduced instance where some goods are removed.
Our first two lemmas, at a high level, establish the relation between agents' MMS values in a reduced instance and those in the original instance.

\begin{lemma}
    \label{lemma:reduced instance ok n agents}
    Let $I = \langle \agents = [n], \divGoods \rangle$ be an instance, $\divGoods^r \subseteq \divGoods$ be a subset of goods, and $I' = \langle \agents, \divGoods \setminus \divGoods^r \rangle$ be an instance where goods~$\divGoods^r$ are removed from~$\divGoods$.
    Then, for any agent~$i \in \agents$,
    \[
        \MMS_i(n, \divGoods \setminus \divGoods^r) + v_i(\divGoods^r) \geq \MMS_i(n, \divGoods).
    \]
\end{lemma}

\begin{proof}
Denote by~$\divGoods' \coloneqq \divGoods \setminus \divGoods^r$ the goods in instance~$I'$.
Fix any agent~$i \in \agents$.
Let $(P_1, P_2, \dots, P_n)$ be an MMS partition of agent~$i$ for goods~$\divGoods$ in instance~$I$, i.e., $\MMS_i(n, \divGoods) = \min_{j \in [n]} v_i(P_j)$.
For all~$j \in [n]$, each good (or a fraction of it) in bundle~$P_j$ must be from either~$\divGoods'$ or~$\divGoods^r$.
Let $P_j = P'_j \cup P^r_j$, where $P'_j = P_j \cap \divGoods'$ and $\divGoods^r_j = \divGoods_j \cap \divGoods^r$.
Since agents are assumed to have additive utilities across goods and~$\divGoods'$ and~$\divGoods^r$ are disjoint, we have
\[
v_i(P_j) = v_i(P'_j \cup P^r_j) = v_i(P'_j) + v_i(P^r_j).
\]
Moreover, since $v_i(P^r_j) \leq v_i(\divGoods^r)$, we have
\[
v_i(P_j) = v_i(P'_j) + v_i(P^r_j) \leq v_i(P'_j) + v_i(\divGoods^r).
\]
We are now ready to establish the inequality stated in the \namecref{lemma:reduced instance ok n agents}:
\begin{align*}
\MMS_i(n, \divGoods) & = \min_{j \in [n]} v_i(P_j) = \min_{j \in [n]} \big( v_i(P'_j) + v_i(P^r_j) \big) \\
& \leq \min_{j \in [n]} (v_i(P'_j) + v_i(\divGoods^r)) \\
& = v_i(\divGoods^r) + \min_{j \in [n]} v_i(P'_j) \\
& \leq v_i(\divGoods^r) + \max_{(P'_1, \dots, P'_n) \in \Pi_n(\divGoods')} \min_{j \in [n]} v_i(P'_j) \\
& = v_i(\divGoods^r) + \MMS_i(n, \divGoods'),
\end{align*}
where the last inequality follows from $\divGoods' = \bigcup_{j \in [n]} P'_j$.
\end{proof}

\begin{lemma}
    \label{lemma:reduced instance MMS}
    Consider any instance $\langle [n], \divGoods \rangle$ and fix~$i \in [n]$.
    Let $\divGoods' \subseteq \divGoods$ be an integral subset of goods such that $v_i(\divGoods') \leq \alpha \cdot \MMS_i(n, \divGoods)$.
    Then, $\MMS_i(n-1, \divGoods \setminus \divGoods') \geq (1 - \alpha) \cdot \MMS_i(n, \divGoods)$.
\end{lemma}

\begin{proof}
Consider agent~$i$'s MMS partition $\mathcal{P} = (P_1, \dots, P_n)$ of goods~$\divGoods$.
Since~$\divGoods'$ contains goods integrally, it contributes at most $\alpha \cdot \MMS_i(n, \divGoods)$ to each of the~$n$ bundles in~$\mathcal{P}$.
Therefore, after removing all goods in~$\divGoods'$, we have an $n$-partition of goods~$\divGoods \setminus \divGoods'$ such that all~$n$ bundles are valued at least $(1 - \alpha) \cdot \MMS_i(n, \divGoods)$.
Combining any two bundles would result in an $(n-1)$-partition of~$\divGoods \setminus \divGoods'$ such that all~$n-1$ bundles are valued at least $(1-\alpha) \cdot \MMS_i(n, \divGoods)$, meaning that $\MMS_i(n-1, \divGoods \setminus \divGoods') \geq (1 - \alpha) \cdot \MMS_i(n, \divGoods)$.
\end{proof}

Our next lemma shows that if we remove a bundle of goods whose thresholds are greater than~$0.5$ (for some agent) and whose value is upper bounded, then the agent still have enough value for the remaining goods.

\begin{lemma}
    \label{lemma:small goods partition generalisation}
    Consider an instance $\langle [n], \divGoods \rangle$ and let~$i \in [n]$.
    If there exists $G \subseteq \divGoods$ such that $c_{i, g} > 0.5$ for all~$g \in G$ and $v_i(G) < \alpha \cdot \MMS_i(n, \divGoods)$ where $\alpha \geq 0$ is some non-negative real number, then $v_i(\divGoods \setminus G) \geq \frac{n - \alpha}{n} \cdot \MMS_i(n, \divGoods)$.
\end{lemma}

\begin{proof}
Suppose for the sake of contradiction that $v_i(\divGoods \setminus G) < \frac{n - \alpha}{n} \cdot \MMS_i(n, \divGoods)$.
Consider $I' = \langle N, G \rangle$ where goods~$\divGoods \setminus G$ are removed.
Since all goods in~$G$ have threshold larger than~$0.5$ for agent~$i$, these goods will not be split between multiple bundles in agent~$i$'s MMS partition for goods~$G$ and thus, $\MMS_i(n, G) \leq v_i(G) / n$.
Due to the assumption in the \namecref{lemma:small goods partition generalisation} statement, we have $\MMS_i(n, G) < \frac{\alpha}{n} \cdot \MMS_i(n, \divGoods)$.
Finally, by \Cref{lemma:reduced instance ok n agents},
\[
\MMS_i(n, \divGoods) \leq \MMS_i(n, G) + v_i(\divGoods \setminus G) < \frac{\alpha}{n} \cdot \MMS_i(n, \divGoods) + \frac{n - \alpha}{n} \cdot \MMS_i(n, \divGoods) = \MMS_i(n, \divGoods),
\]
a contradiction, as desired.
\end{proof}

We now introduce the concepts of \emph{compatible goods} and \emph{contested goods}.
Given any instance $\langle [n], \divGoods \rangle$, for each~$g \in \divGoods$, good~$g$ is said to be a \emph{compatible good} if $\sum_{i \in [n]} c_{i, g} \leq 1$; otherwise, the good~$g$ is said to be a \emph{contested good}.
In words, compatible goods are those where the thresholds for all agents sum to at most~$1$.
Intuitively, each compatible good~$g$ can be allocated to all agents and each agent~$i$ gets a utility of~$v_i(g)$.
By applying \Cref{lemma:reduced instance ok n agents} to compatible goods, we formalize below the intuition that it suffices to focus on allocating contested goods.

Given any instance $I = \langle \agents = [n], \divGoods \rangle$, let~$G^c$ (resp., $G$) be the set of all compatible (resp., contested) goods in instance~$I$ such that $\divGoods = G^c \cup G$.
Consider a reduced instance $I' = \langle \agents, G \rangle$ with the same set of agents~$\agents$ and all compatible goods~$G^c$ being removed.
Suppose we are given an $\alpha$-MMS allocation $\alloc' = (A'_1, A'_2, \dots, A'_n)$ for instance~$I'$, where~$\alpha \in [0, 1]$.
Put differently, for all agents~$i \in N$,
\[
    v_i(A'_i) \geq \alpha \cdot \MMS_i(n, G).
\]
We construct an allocation~$\alloc = (A_1, A_2, \dots, A_n)$ as follows:
\begin{itemize}
    \item for each~$g \in G$ and $i \in \agents$, let $A_{i, g} = A'_{i, g}$; and
    \item for each~$g \in G^c$, let $A_{i, g} = c_{i, g}$ for all~$i \in [n-1]$ and $A_{n, g} = 1 - \sum_{i \in [n-1]} c_{i, g}$.
\end{itemize}
It can be verified that for each~$i \in N$,
\begin{align*}
v_i(A_i) = v_i(A'_i) + \sum_{g \in G^c} v_i(g) = v_i(A'_i) + v_i(G^c) & \geq \alpha \cdot \MMS_i(n, G) + v_i(G^c) \\
& \geq \alpha \cdot \big( \MMS_i(n, G) + v_i(G^c) \big) \\
& \geq \alpha \cdot \MMS_i(n, \divGoods),
\end{align*}
where the last transition is due to \Cref{lemma:reduced instance ok n agents}.
As a result, allocation~$\alloc$ is an $\alpha$-MMS allocation for instance~$I$.
To summarize, we have the following \namecref{corollary:remove compatible}.

\begin{lemma}
    \label{corollary:remove compatible}
    Let~$G$ be the set of all contested goods in $I = \langle \agents, \divGoods \rangle$.
    For any~$\alpha \in [0, 1]$, if there exists an $\alpha$-MMS allocation for instance $\langle N, G \rangle$, an $\alpha$-MMS allocation for instance~$I$ always exists.
\end{lemma}

\subsubsection{Two Agents}
\label{sec:MMS:piecewise-const:2-agent}

\begin{algorithm}[t]
    \caption{\twoAgentOneHalfMMS$([2], \divGoods)$}
    \label{alg:2-agent-1/2MMS}
    \DontPrintSemicolon

    \KwIn{An instance~$\langle \agents = [2], \divGoods \rangle$, where both agents have one-breakpoint piecewise-constant valuations.}
    \KwOut{A $\frac{1}{2}$-MMS allocation~$\alloc$.}

    $A_1, A_2 \gets \emptyset$\;

    \tcp{Allocate compatible goods.}

    $G^c \gets \{g \in \divGoods \colon c_{1, g} + c_{2, g} \leq 1\}$\;
    \lForEach{$g \in G^c$}{ \label{alg:2-agent-1/2MMS:compatible-goods-alloc}
        $A_{1, g} \gets c_{1, g}$; $A_{2, g} \gets 1 - c_{1, g}$ 
    }

    \tcp{Allocate the remaining contested goods.}

    $G \gets \divGoods \setminus G^c$\;
    \lForEach{$i \in \agents$}{
        Compute $\MMS_i \coloneqq \MMS_i(2, G)$.
    }

    \eIf{$\exists g^* \in G, i \in \agents$ such that $v_i(g^*) \geq \MMS_i / 2$}{ \label{alg:2-agent-1/2MMS:large-good-if}
    \If{$v_{3-i}(g^*) \geq \MMS_{3-i} / 2$}{
        Relabel~$i$ as the agent such that $c_{i, g^*} \leq c_{3-i, g^*}$.
    }
    $A_i \gets A_i \cup \{g^*\}$; $A_{3-i} \gets A_{3-i} \cup \big( G \setminus \{g^*\} \big)$\; \label{alg:2-agent-1/2MMS:large-good-alloc}
    }{
    \lForEach{$i \in \agents$}{
        $G_i \gets \{g \in G \colon c_{i, g} > 0.5\}$
    }

    \uIf{$\exists i \in \agents$ such that $v_i(G_i) \geq \MMS_i$}{
    Add one good integrally at a time from~$G_i$ to an empty bundle~$G^*$ until $v_i(G^*) \in [\MMS_i / 2, \MMS_i)$.\;
    Let agent~$3-i$ pick her preferred bundle out of~$G^*$ and $G \setminus G^*$, and allocate the other bundle to~$i$.\;
    }\lElse{
        $A_i \gets A_i \cup \big( G \setminus G_i \big)$; $A_{3-i} \gets A_{3-i} \cup G_i$
    }
    }

    \Return{$\alloc = (A_1, A_2)$}
\end{algorithm}

We are now ready to present an algorithm which always returns a $\frac{1}{2}$-MMS allocation for two agents having one-breakpoint piecewise-constant valuations.
The pseudocode can be found in \Cref{alg:2-agent-1/2MMS}.
In this algorithm, we first allocate all compatible goods~$G^c$ with the respective thresholds to both agents, followed by allocating the remaining contested goods $G$ by breaking into cases based on whether there exists some (contested) good~$g \in G$ and agent~$i$ such that $v_i(g) > \MMS_i(2, G) / 2$.
We demonstrate \Cref{alg:2-agent-1/2MMS} in \Cref{ex:2-agent}.

\begin{example}[Demonstration of \Cref{alg:2-agent-1/2MMS}]
\label{ex:2-agent}
Consider an instance with $n=2$ agents and $m=5$ goods, with the valuation functions of both agents to all goods being piecewise-constant with 1 breakpoint as follows.
\begin{figure}[!h]
\begin{subfigure}[b]{.18\linewidth}
\centering
\begin{tikzpicture}
                    \draw[->] (0,0) -- (2,0) node[right] {$p$};
                    \draw[->] (0,0) -- (0,2) node[above] {$v_1(p)$};
                    \draw (0.3,2pt) -- ++ (0,-4pt) node[below] {$0.15$};
                    \draw (2pt,0.7) -- ++ (-4pt,0) node[left]  {$0.35$};
                    \draw[domain=0:0.3] plot (\x, {0});
                    \draw[domain=0.3:2] plot (\x, {0.7});
                    \node[draw, circle, fill, inner sep=1pt] at (0.3,0.7) {};
                    \node[draw, circle, inner sep=1pt] at (0.3,0) {};
                \end{tikzpicture}
                \caption{$v_1(g_1)$}
            \end{subfigure}
            \hfill
            \begin{subfigure}[b]{.18\linewidth}
                \centering
                \begin{tikzpicture}
                    \draw[->] (0,0) -- (2,0) node[right] {$p$};
                    \draw[->] (0,0) -- (0,2) node[above] {$v_2(p)$};
                    \draw (0.6,2pt) -- ++ (0,-4pt) node[below] {$0.3$};
                    \draw (2pt,0.6) -- ++ (-4pt,0) node[left]  {$0.3$};
                    \draw[domain=0:0.6] plot (\x, {0});
                    \draw[domain=0.6:2] plot (\x, {0.6});
                    \node[draw, circle, fill, inner sep=1pt] at (0.6,0.6) {};
                    \node[draw, circle, inner sep=1pt] at (0.6,0) {};
                \end{tikzpicture}
                \caption{$v_1(g_2)$}
            \end{subfigure}
            \hfill
            \begin{subfigure}[b]{.18\linewidth}
                \centering
                \begin{tikzpicture}
                    \draw[->] (0,0) -- (2,0) node[right] {$p$};
                    \draw[->] (0,0) -- (0,2) node[above] {$v_3(p)$};
                    \draw (1.5,2pt) -- ++ (0,-4pt) node[below] {$0.75$};
                    \draw (2pt,0.8) -- ++ (-4pt,0) node[left]  {$0.4$};
                    \draw[domain=0:1.5] plot (\x, {0});
                    \draw[domain=1.5:2] plot (\x, {0.8});
                    \node[draw, circle, fill, inner sep=1pt] at (1.5,0.8) {};
                    \node[draw, circle, inner sep=1pt] at (1.5,0) {};
                \end{tikzpicture}
                \caption{$v_1(g_3)$}
            \end{subfigure}
            \hfill
            \begin{subfigure}[b]{.18\linewidth}
                \centering
                \begin{tikzpicture}
                    \draw[->] (0,0) -- (2,0) node[right] {$p$};
                    \draw[->] (0,0) -- (0,2) node[above] {$v_4(p)$};
                    \draw (0.4,2pt) -- ++ (0,-4pt) node[below] {$0.2$};
                    \draw (2pt,0.9) -- ++ (-4pt,0) node[left]  {$0.45$};
                    \draw[domain=0:0.4] plot (\x, {0});
                    \draw[domain=0.4:2] plot (\x, {0.9});
                    \node[draw, circle, fill, inner sep=1pt] at (0.4,0.9) {};
                    \node[draw, circle, inner sep=1pt] at (0.4,0) {};
                \end{tikzpicture}
                \caption{$v_1(g_4)$}
            \end{subfigure}
            \hfill
            \begin{subfigure}[b]{.18\linewidth}
                \centering
                \begin{tikzpicture}
                    \draw[->] (0,0) -- (2,0) node[right] {$p$};
                    \draw[->] (0,0) -- (0,2) node[above] {$v_5(p)$};
                    \draw (1.6,2pt) -- ++ (0,-4pt) node[below] {$0.8$};
                    \draw (2pt,0.5) -- ++ (-4pt,0) node[left]  {$0.25$};
                    \draw[domain=0:1.6] plot (\x, {0});
                    \draw[domain=1.6:2] plot (\x, {0.5});
                    \node[draw, circle, fill, inner sep=1pt] at (1.6,0.5) {};
                    \node[draw, circle, inner sep=1pt] at (1.6,0) {};
                \end{tikzpicture}
                \caption{$v_1(g_5)$}
            \end{subfigure}
            \begin{subfigure}[b]{.18\linewidth}
                \centering
                \begin{tikzpicture}
                    \draw[->] (0,0) -- (2,0) node[right] {$p$};
                    \draw[->] (0,0) -- (0,2) node[above] {$v_1(p)$};
                    \draw (1.2,2pt) -- ++ (0,-4pt) node[below] {$0.6$};
                    \draw (2pt,0.2) -- ++ (-4pt,0) node[left]  {$0.1$};
                    \draw[domain=0:1.2] plot (\x, {0});
                    \draw[domain=1.2:2] plot (\x, {0.2});
                    \node[draw, circle, fill, inner sep=1pt] at (1.2,0.2) {};
                    \node[draw, circle, inner sep=1pt] at (1.2,0) {};
                \end{tikzpicture}
                \caption{$v_2(g_1)$}
            \end{subfigure}
            \hfill
            \begin{subfigure}[b]{.18\linewidth}
                \centering
                \begin{tikzpicture}
                    \draw[->] (0,0) -- (2,0) node[right] {$p$};
                    \draw[->] (0,0) -- (0,2) node[above] {$v_2(p)$};
                    \draw (1.8,2pt) -- ++ (0,-4pt) node[below] {$0.9$};
                    \draw (2pt,0.7) -- ++ (-4pt,0) node[left]  {$0.35$};
                    \draw[domain=0:1.8] plot (\x, {0});
                    \draw[domain=1.8:2] plot (\x, {0.7});
                    \node[draw, circle, fill, inner sep=1pt] at (1.8,0.7) {};
                    \node[draw, circle, inner sep=1pt] at (1.8,0) {};
                \end{tikzpicture}
                \caption{$v_2(g_2)$}
            \end{subfigure}
            \hfill
            \begin{subfigure}[b]{.18\linewidth}
                \centering
                \begin{tikzpicture}
                    \draw[->] (0,0) -- (2,0) node[right] {$p$};
                    \draw[->] (0,0) -- (0,2) node[above] {$v_3(p)$};
                    \draw (0.8,2pt) -- ++ (0,-4pt) node[below] {$0.4$};
                    \draw (2pt,0.4) -- ++ (-4pt,0) node[left]  {$0.2$};
                    \draw[domain=0:0.8] plot (\x, {0});
                    \draw[domain=0.8:2] plot (\x, {0.4});
                    \node[draw, circle, fill, inner sep=1pt] at (0.8,0.4) {};
                    \node[draw, circle, inner sep=1pt] at (0.8,0) {};
                \end{tikzpicture}
                \caption{$v_2(g_3)$}
            \end{subfigure}
            \hfill
            \begin{subfigure}[b]{.18\linewidth}
                \centering
                \begin{tikzpicture}
                    \draw[->] (0,0) -- (2,0) node[right] {$p$};
                    \draw[->] (0,0) -- (0,2) node[above] {$v_4(p)$};
                    \draw (1.7,2pt) -- ++ (0,-4pt) node[below] {$0.85$};
                    \draw (2pt,0.8) -- ++ (-4pt,0) node[left]  {$0.4$};
                    \draw[domain=0:1.7] plot (\x, {0});
                    \draw[domain=1.7:2] plot (\x, {0.8});
                    \node[draw, circle, fill, inner sep=1pt] at (1.7,0.8) {};
                    \node[draw, circle, inner sep=1pt] at (1.7,0) {};
                \end{tikzpicture}
                \caption{$v_2(g_4)$}
            \end{subfigure}
            \hfill
            \begin{subfigure}[b]{.18\linewidth}
                \centering
                \begin{tikzpicture}
                    \draw[->] (0,0) -- (2,0) node[right] {$p$};
                    \draw[->] (0,0) -- (0,2) node[above] {$v_5(p)$};
                    \draw (0.6,2pt) -- ++ (0,-4pt) node[below] {$0.3$};
                    \draw (2pt,0.9) -- ++ (-4pt,0) node[left]  {$0.45$};
                    \draw[domain=0:0.6] plot (\x, {0});
                    \draw[domain=0.6:2] plot (\x, {0.9});
                    \node[draw, circle, fill, inner sep=1pt] at (0.6,0.9) {};
                    \node[draw, circle, inner sep=1pt] at (0.6,0) {};
                \end{tikzpicture}
                \caption{$v_2(g_5)$}
            \end{subfigure}

            \caption{Valution functions $v_i(g_j)$ for agents $i\in[2]$ and goods $j\in[5]$.}\label{fig:example-two-agent-1/2-MMS-piecewise-constant}
        \end{figure}
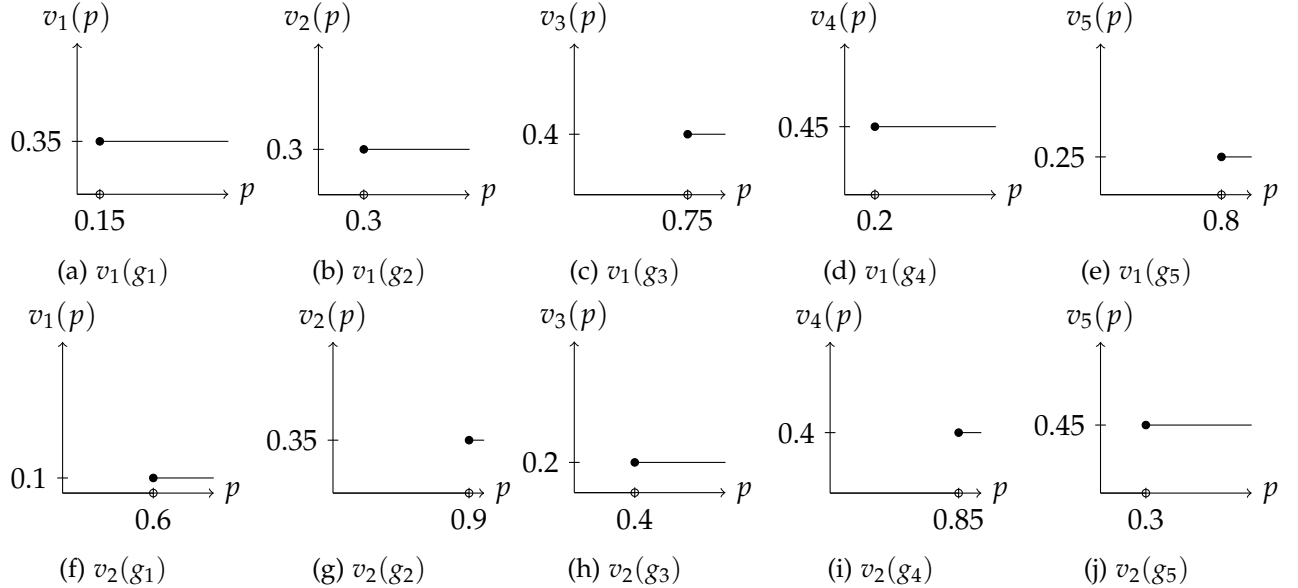

        It can be checked that the MMS of the two agents are $1.35$ and $1.05$ respectively, and we proceed to find the $1/2$-MMS allocation by \Cref{alg:2-agent-1/2MMS}. Firstly, note that $g_1$ is compatible and we allocate $0.15\cdot g_1$ to agent $1$ and $0.85\cdot g_1$ to agent $2$. Then, we observe that the MMS of both agents is $1$~under the reduced instance $I'=([2], \{g_2, g_3, g_4, g_5\})$. Further, no goods are valued larger than $0.5\,\MMS$ for both agents under the reduced instance $I'$, and so we find $G_1 = \{g_3, g_5\}$, $G_2 = \{g_2, g_4\}$. Clearly, $v_1(G_1) = 0.65 \geq 1/2\MMS_1$, and thus by iteratively adding goods from $G_1$ to $G^*$ we find that $G^* = \{g_3, g_5\}$, and $v_1(G^*) \geq 1/2\MMS_1$. Finally, we have that $v_2(G^*) = 0.65$ and $v_2(\divGoods\setminus G^*) = 0.75$, and so we allocate $\divGoods\setminus G^*$ to agent $2$ and $G^*$ to agent $1$.

        Therefore, the final allocation would be $A_1 = \{0.15\cdot g_1, g_3, g_5\}$ and $A_2 = \{0.85\cdot g_1, g_2, g_4\}$. One can check that $v_2(A_1) = 1$ and $v_2(A_2) = 0.85$, and the allocation $A = \langle A_1, A_2\rangle$ is a $1/2$-MMS allocation.
    \end{example}

\begin{theorem}
    \label{thm:2-agent:1/2-MMS-one-breakpoint}
    For $n = 2$ agents and one-breakpoint piecewise-constant valuations, \Cref{alg:2-agent-1/2MMS} returns a $\frac{1}{2}$-MMS allocation.
\end{theorem}

\begin{proof}
    \Cref{alg:2-agent-1/2MMS} starts by allocating all compatible goods~$G^c$ in \cref{alg:2-agent-1/2MMS:compatible-goods-alloc}.
    By \Cref{corollary:remove compatible}, it suffices to show that \Cref{alg:2-agent-1/2MMS} finds a $1/2$-MMS allocation in the reduced instance $I' = \langle N, G = \divGoods \setminus G^c\rangle$ which contains only contested goods.
    For notational convenience, in the remainder of this proof, let $\MMS_i$ denote $\MMS_i(\agents, G)$, i.e., the maximin share in the reduced instance~$I'$.

    We distinguish cases based on whether there exists good~$g^* \in G$ and~$i \in \agents$ such that $v_i(g^*) \geq \MMS_i / 2$.
    Suppose that the \verb|if|-statement in \cref{alg:2-agent-1/2MMS:large-good-if} is evaluated as true, i.e., such a good~$g^*$ and agent~$i$ exists.
    We first consider the case where $v_{3-i}(g^*) \geq \frac{\MMS_{3-i}}{2}$, and assume without loss of generality that $c_{i, g^*} \leq c_{3-i, g^*}$ (relabel the agents if needed).
    Since good~$g^*$ is contested, $c_{3-i, g^*} > 0.5$, which implies that $v_{3-i}(G \setminus \{g^*\}) \geq \MMS_{3-i}$.
    We now consider the other case where $v_{3-i}(g^*) < \frac{\MMS_{3-i}}{2}$.
    It can be verified that
    \[\textstyle
        v_{3-i}(G \setminus \{g^*\}) = v_{3-i}(G) - v_{3-i}(g^*) \geq \MMS_{3-i} - \frac{\MMS_{3-i}}{2} = \frac{\MMS_{3-i}}{2}.
    \]
    In either case, allocating good~$g^*$ to agent~$i$ and all remaining goods~$G \setminus \{g^*\}$ to the other agent in \cref{alg:2-agent-1/2MMS:large-good-alloc} gives a $\frac{1}{2}$-MMS allocation for instance~$I'$, as desired.

    Finally, we consider the scenario where for all~$i \in \agents$ and~$g \in G$, $v_i(g) < \MMS_i / 2$.
    For each~$i \in N$, let $G_i \subseteq G$ consist of goods where agent~$i$ has thresholds larger than~$0.5$.
    There are two cases:
    \begin{enumerate}[label=(\roman*)]
        \item $v_i(G_i) \geq \MMS_i$ for some agent~$i \in N$, and
        \item $v_i(G_i) < \MMS_i$ for both agents.
    \end{enumerate}
    In the first case, because all goods~$g \in G$ follow that $v_i(g) < \MMS_i / 2$, we can keep adding goods~$g \in G_i$ to find a subset $G^* \subseteq G_i$ such that $v_i(G^*) \in \big[ \frac{\MMS_i}{2}, MMS_i \big)$.
    By \Cref{lemma:small goods partition generalisation}, $v_i(G \setminus G^*) \geq \frac{1}{2} \cdot \MMS_i$.
    Therefore, allocating the preferred bundle between~$G^*$ and~$G \setminus G^*$ to agent~$3-i$ and the other bundle to agent~$i$ is a $1/2$-MMS allocation for instance~$I'$.
    In the second case, we have $v_i(G_i) < \MMS_i$ for both agents~$i \in N$.
    By \Cref{lemma:small goods partition generalisation}, this implies $v_i(G \setminus G_i) \geq \MMS_i / 2$, where $G \setminus G_i$ is the set of goods such that agent~$i$ has thresholds at most~$0.5$.
    However, since all goods in~$G$ are contested, the sets $G \setminus G_i$ and $G \setminus G_{3-i}$ are mutually exclusive.
    This means that allocating $G \setminus G_i$ to agent~$i$ and all remaining goods to the other agent satisfies $1/2$-MMS in instance~$I'$.

    In conclusion, \Cref{alg:2-agent-1/2MMS} finds a $1/2$-MMS allocation in the reduced instance~$I'$.
    Moreover, due to \Cref{corollary:remove compatible}, \Cref{alg:2-agent-1/2MMS} finds a $\frac{1}{2}$-MMS allocation for the two agents in the original instance.
\end{proof}

\subsubsection{Three Agents}
\label{sec:MMS:piecewise-const:3-agent}

Built upon our $2$-agent-$\frac{1}{2}$-MMS allocation algorithm, we proceed to present an algorithm that finds a $1/3$-MMS allocation for three agents with one-breakpoint piecewise-constant valuations.
The pseudocode is given as \Cref{alg:3-agent-1/3-MMS}.

\Cref{alg:3-agent-1/3-MMS} takes as input a set of agents~$\agents = [3]$ and a set of goods~$\divGoods = G^c \cup G$, where $G^c$ denotes the set of compatible goods and $G$ denotes the set of contested goods.
Following \Cref{corollary:remove compatible}, \Cref{alg:3-agent-1/3-MMS} starts by allocating all compatible goods~$G^c$ with the respective thresholds to the three agents.
It remains to show that our algorithm finds a $\frac{1}{3}$-MMS allocation in the reduced instance with agents~$\agents$ and contested goods~$G$.
Note also that in the following, for each agent~$i \in \agents$, $\MMS_i \coloneqq \MMS_i(3, G)$.

Based on whether there exists some agent~$i$ and good~$g \in G$ such that $v_i(g) \geq \MMS_i / 3$, our algorithm can be naturally broken into two components, with \crefrange{alg:3-agent-1/3-MMS:large-goods-begins}{alg:3-agent-1/3-MMS:large-goods-ends} processing instances with some large good and \crefrange{alg:3-agent-1/3-MMS:small-goods-begins}{alg:3-agent-1/3-MMS:small-goods-ends} processing remaining instances with only small goods.
Comparing with the $2$-agent case presented previously, our \Cref{alg:3-agent-1/3-MMS} is more intricate in the sense that besides agents' valuations for a single good or a set of goods, we also need to reason about the thresholds by which agents start having positive utilities for a given good.
Naturally, this is more complicated since each good can be allocated to any of the 7 non-empty subsets of agents $\{1, 2, 3\}$ based on the instance.

\begin{algorithm}[!htbp]
    \caption{$3$-Agent-$\frac{1}{3}$-MMS Allocation}
    \label{alg:3-agent-1/3-MMS}
    \DontPrintSemicolon

    \KwIn{Instance $I = \langle \agents = [3], \divGoods = G^c \cup G \rangle$, where agents have one-breakpont piecewise-constant valuations.}
    \KwOut{A $1/3$-MMS allocation~$\alloc$.}

    $A_1, A_2, A_3 \gets \emptyset$\;

    \lForEach{$g \in G^c$}{
        $A_{i, g} \gets c_{i, g} \forall i \in [n-1]$; $A_{n, g} \gets 1 - \sum_{i \in [n-1]} c_{i, g}$
    }

    \lForEach{$i \in \agents$}{
        Compute $\MMS_i \coloneqq \MMS_i(3, G)$.
    }

    \tcp{Process instances with some large goods.}
    \If{$\exists g^* \in G$ and $a \in \agents$ such that $v_a(g^*) \geq \frac{\MMS_a}{3}$}{ \label{alg:3-agent-1/3-MMS:large-goods-begins}
        Let~$S \subsetneq \agents$ be the set of agents such that $v_i(g^*) \geq \frac{\MMS_i}{3}$ for all~$i \in S$ and $\sum_{j \in S} c_{j, g^*} \leq 1$, breaking ties in favour of larger~$|S|$ and then smaller $\sum_{j \in S} c_{j, g^*}$.\;
        \lForEach{$i \in S$}{
            $A_{i, g^*} \gets c_{i, g^*}$
        }
        \leIf{$|S| = 2$}{
            $A_k \gets A_k \cup G \setminus \{g^*\}$, where $k \notin S$.\;
        }{
            Call \Cref{alg:2-agent-1/2MMS} on sub-instance~$\langle \agents \setminus S, G \setminus \{g^*\} \rangle$
        }
        \Return{$\alloc = (A_1, A_2, A_3)$} \label{alg:3-agent-1/3-MMS:large-goods-ends}
    }

    \tcp{Allocate small goods. Note, $\forall i \in \agents, v_i(g) < \frac{\MMS_i}{3}.$}
    \lForEach{$i \in \agents$}{ \label{alg:3-agent-1/3-MMS:small-goods-begins}
        $G_i \gets \{g \in G \colon c_{i, g} > 0.5\}$; $G'_i \gets G \setminus G_i$
    }

    \If{$\exists i \in \agents$ such that $v_i(G_i) \geq \MMS_i$}{
        Partitiong~$G$ into three integral bundles $G_1^*, G_2^*, G_3^*$ s.t. $v_i(G_1^*), v_i(G_2^*), v_i(G_3^*) \geq \frac{\MMS_i}{3}$.\;
        \eIf{$\exists j \neq i$ values two distinct bundles at least~$\frac{\MMS_j}{3}$}{
            \Return{the preferred bundle to~$k \in \agents \setminus \{i, j\}$, then the preferred bundle to~$j$, and finally the last bundle to~$i$.}
        }{
            Let~$G_a^*$ be the bundle s.t.\ $v_j(G_a^*) < \frac{\MMS_j}{3}$ for all~$j \neq i$.\;
            $A_i \gets G_a^*$\;
            \Return{\twoAgentOneHalfMMS$([3] \setminus \{i\}, G \setminus G^*_a)$}
        }
    }

    $\forall i, j \in \agents, G'_{\{i, j\}} \gets \{g \in G \colon c_{i, g} \leq 0.5 \land c_{j, g} \leq 0.5\}$\;
    \eIf{$\exists G'_{\{i, j\}}$ s.t.\ $v_i(G'_{\{i, j\}}) \geq \frac{\MMS_i}{3}$ and $v_j(G'_{\{i, j\}}) \geq \frac{\MMS_j}{3}$}{
    After relabelling the agents, let $G^* \subseteq G'_{\{i, j\}}$ be the bundle such that $v_i(G^*) \geq \frac{\MMS_i}{3}$, $v_j(G^*) \geq \frac{\MMS_j}{3}$, and $v_k(G^*) < \frac{2 \cdot \MMS_k}{3}$.\;
    $\forall g \in G^*$, $A_{i, g} \gets 0.5$ and $A_{j, g} \gets 0.5$; $A_k \gets G \setminus G^*$\;
    }{
    Pick any~$G'_{\{i, k\}}$ such that $v_k(G'_{\{i, k\}}) < \MMS_k / 3$.\;
    Iteratively move~$g \in G'_i$ to~$G^*_i$ until $v_i(G^*_i) \geq \MMS_i / 3$.\;
    $\forall g \in G^*_i$, $A_{i, g} \gets 0.5$ and $A_{j, g} \gets 0.5$\;
    Iteratively move~$g\in G'_j$ to~$G^*_j$ until $v_j(A_j \cup G^*_j) \geq \frac{\MMS_j}{3}$.\;
    $\forall g \in G^*_j$, $A_{j, g} \gets 0.5$ and $A_{k, g} \gets 0.5$\;
    Allocate everything else to~$k$.\;
    }

    \Return{$\alloc = (A_1, A_2, A_3)$} \label{alg:3-agent-1/3-MMS:small-goods-ends}
\end{algorithm}

\begin{theorem}
    \label{thm:3-agent:1/3-MMS-one-breakpoint}
    For $n = 3$ agents and one-breakpoint piecewise-constant valuations, \Cref{alg:3-agent-1/3-MMS} returns a $\frac{1}{3}$-MMS allocation.
\end{theorem}

Before showing \Cref{alg:3-agent-1/3-MMS} indeed finds a $1/3$-MMS allocation, we prove a few useful results.
The following two lemmas allow us to form an allocation when there is some ``large'' good(s) in the given instance, i.e., an agent~$i$ values a good~$g$ at least $\MMS_i / 3$.

\begin{lemma}
    \label{lemma:breakpoint larger than 1/3}
    Let~$i \in \agents$.
    If there exists~$g^* \in G$ such that $c_{i, g^*} > 1/3$, then $v_i(G \setminus g^*) \geq \MMS_i$.
\end{lemma}

\begin{proof}
Suppose that $v_i(G \setminus g^*) < \MMS_i$ and consider an MMS partition $(P_1, P_2, P_3)$ of agent~$i$ for goods~$G$.
Since $c_{i, g^*} > 1/3$, it follows that~$g^*$ can contribute value~$v_i(g^*)$ to at most two bundles in the MMS partition.
However, this would imply that the last bundle in the MMS partition (say~$P_3$) must have value $v_i(P_3) \leq v_i(G \setminus g^*) < \MMS_i$, a contradiction.
\end{proof}

\begin{lemma}
    \label{lemma:breakpoint larger than 1/2}
    Let~$i \in \agents$.
    If there exists~$g^* \in G$ such that $c_{i, g^*} > 1/2$, then $\MMS_i(2, G \setminus \{g^*\}) \geq \MMS_i(3, G)$.
\end{lemma}

\begin{proof}
Consider an MMS partition $(P_1, P_2, P_3)$ of agent~$i$ for goods~$G$.
Since $c_{i, g^*} > 1/2$, good~$g^*$ can contribute value~$v_i(g^*)$ to at most one bundle in the MMS partition.
This implies that after removing~$g^*$ from~$G$, there must be at least two bundles (say, $P_2, P_3$) such that $\min \{v_i(P_2 \setminus \{g^*\}), v_i(P_3 \setminus \{g^*\})\} \geq \MMS_i(3, G)$.
It follows that $\MMS_i(2, G \setminus \{g^*\}) \geq \MMS_i(3, G)$.
\end{proof}

We are now ready establish the correctness of \crefrange{alg:3-agent-1/3-MMS:large-goods-begins}{alg:3-agent-1/3-MMS:large-goods-ends} of \Cref{alg:3-agent-1/3-MMS}.

\begin{lemma}
    \label{lem:3-agent:1/3-MMS-one-breakpoint:large-goods}
    Given agents~$\agents = [3]$ and (contested) goods~$G$, if there exists some~$a \in \agents$ and some~$g^* \in G$ such that $v_a(g^*) \geq \frac{\MMS_a}{3}$, \crefrange{alg:3-agent-1/3-MMS:large-goods-begins}{alg:3-agent-1/3-MMS:large-goods-ends} of \Cref{alg:3-agent-1/3-MMS} returns a $\frac{1}{3}$-MMS allocation.
\end{lemma}

\begin{proof}
    Let $S \subseteq \agents$ be the set of agents such that $v_i(g^*) \geq \frac{\MMS_i}{3}$ for all~$i \in S$ and $\sum_{j \in S} c_{j, g^*} \leq 1$, breaking ties in favour of larger~$|S|$ and then smaller~$\sum_{j \in S} c_{j, g^*}$.
    In other words, we can allocate a fraction~$c_{i, g^*}$ of good~$g^*$ to each agent~$i \in S$ and the agent receives a utility of at least~$\frac{\MMS_i}{3}$.
    Recall that we are allocating contested goods, meaning that good~$g^*$ can be allocated to at most two agents.
    Clearly, $|S| \geq 1$, as the \verb|if|-condition in \cref{alg:3-agent-1/3-MMS:large-goods-begins} is evaluated as true.

    Below, we distinguish two cases based on whether $|S| = 2$ or $|S| = 1$.
    When $|S| = 2$, since we break ties in favour of smaller $\sum_{j \in S} c_{j, g^*}$, the agent~$k \notin S$ must either have $c_{k, g^*} > 1/3$, or $v_k(g^*) < \frac{\MMS_k}{3}$.
    In the first case, \Cref{lemma:breakpoint larger than 1/3} implies that $v_k(G \setminus \{g^*\}) \geq \MMS_k$; and in the second case, we have $v_k(G\setminus \{g^*\}) > \MMS_k - \frac{\MMS_k}{3} = \frac{2\MMS_k}{3}$. This implies that allocating $G \setminus \{g^*\}$ would give a $1/3$-MMS allocation.

    When $|S| = 1$, let~$j$ and~$k$ be the agents not in~$S$.
    For any agent $a \in \{j, k\}$, we must have either $c_{a, g^*} > 1/2$ or $v_a(g^*) < \frac{\MMS_a}{3}$.
    In the first case, \Cref{lemma:breakpoint larger than 1/2} implies that under the reduced instance $I' = \langle \{j, k\}, G \setminus \{g^*\} \rangle$, $\MMS_a(2, G \setminus \{g^*\}) \geq \MMS_a(3, G)$; in the second case, since $v_a(g^*) < \frac{\MMS_a}{3}$, \Cref{lemma:reduced instance MMS} implies that $\MMS_a(2, G\setminus \{g^*\}) \geq \frac{2}{3}\MMS_a(3, G)$.
    Therefore in both cases, a $1/2$-MMS allocation under instance~$I'$ must correspond to an at least $1/3$-MMS allocation under the instance with all three agents and goods~$G$.
    As a result, allocating $G \setminus g^*$ to $j$ and $k$ using \Cref{alg:2-agent-1/2MMS} would give a $1/3$-MMS allocation.
    The conclusion follows.
\end{proof}

In the following, we are concerned with $3$-agent instances with only ``small'' goods.
Formally, for all agents~$i \in \agents$ and all goods~$g \in G$, we have $v_i(g) < \MMS_i / 3$.
For each agent~$i \in \agents$, let~$G_i$ (resp., $G_i'$) denote the set of goods~$g$ such that $c_{i, g} > 0.5$ (resp., $c_{i, g} \leq 0.5$).

First, we observe that when $v_i(G_i) \geq \MMS_i$ for some agent~$i$, then we can partition goods~$G$ into three integral bundles such that all three bundles are worth at least~$\MMS_i / 3$.

\begin{lemma}
    \label{lemma:integral partition}
    Let~$i \in [3]$ and suppose that $v_i(g) < \MMS_i / 3$ for all goods~$g \in G$.
    If $v_i(G_i) \geq \MMS_i$, then~$G$ can be partitioned into three integral bundles each of which is worth at least~$\MMS_i / 3$.
\end{lemma}

\begin{proof}
First, consider the case where $v_i(G_i) \geq \frac{5}{3} \cdot \MMS_i$.
Since all goods are valued less than $\MMS_i / 3$, we could iteratively move (integral) goods from~$G_i$ into a bundle~$B$ until $v_i(B) \in \big[ \frac{\MMS_i}{3}, \frac{2 \cdot \MMS_i}{3} \big)$.
Following this manner, we could form two bundles~$B_1, B_2$ such that $\frac{\MMS_i}{3} \leq v_i(B_1), v_i(B_2) < \frac{2  \cdot \MMS_i}{3}$.
Let the set of goods that have not yet been allocated be~$G^*$.
Since $v_i(G_i) \geq \frac{5}{3} \cdot \MMS_i$, it follows that $v_i(G^*) \geq v_i(G_i) - \sum_{j = 1}^{2} v_i(B_j) > \frac{5}{3} \cdot \MMS_i - 2 \cdot \frac{2 \cdot \MMS_i}{3} = \frac{\MMS_i}{3}$.
Adding all goods~$G \setminus G_i$ to the last bundle gives a $3$-partition of~$G$ into integral bundles satisfying the claimed condition.

Next, consider the case where $\MMS_i \leq v_i(G_i) < \frac{5}{3} \cdot \MMS_i$.
Similarly, by iteratively moving goods from~$G_i$ to a bundle~$B$, we could form two bundles~$B$ and~$G_i \setminus B$ such that $v_i(B) \in \big[ \frac{\MMS_i}{3}, \frac{2 \cdot \MMS_i}{3} \big)$ and $v_i(G_i \setminus B) \geq \MMS_i - \frac{2 \cdot \MMS_i}{3} = \frac{\MMS_i}{3}$.
By \Cref{lemma:small goods partition generalisation}, we have $v_i(G \setminus G_i) \geq \frac{3 - 5/3}{3} \cdot \MMS_i = \frac{4}{9} \cdot \MMS_i$.
Then $B$, $G_i \setminus B$ and $G \setminus G_i$ form the desired integral partition of~$G$.
\end{proof}

We are now ready to establish the correctness of \crefrange{alg:3-agent-1/3-MMS:small-goods-begins}{alg:3-agent-1/3-MMS:small-goods-ends} of \Cref{alg:3-agent-1/3-MMS}.

\begin{lemma}
    \label{lem:3-agent:1/3-MMS-one-breakpoint:small-goods}
    Given agents~$\agents = [3]$ and (contested) goods~$G$, if for all goods~$g \in G$ and agents~$i \in \agents$, $v_i(g) < \MMS_i / 3$, \crefrange{alg:3-agent-1/3-MMS:small-goods-begins}{alg:3-agent-1/3-MMS:small-goods-ends}
    \Cref{alg:3-agent-1/3-MMS} returns a $1/3$-MMS allocation.
\end{lemma}

\begin{proof}
We distinguish between two cases based on whether~$G_i$ is valuable enough for some agent~$i \in \agents$.

\paragraph{Case~1: $v_i(G_i) \geq \MMS_i$ for some~$i \in \agents$.}
By \Cref{lemma:integral partition}, goods~$G$ can be partitioned into three integral bundles $G^*_1, G^*_2, G^*_3$ such that all three integral bundles are worth at least~$\MMS_i / 3$.
Observe that for all~$j \in [3]$, since the three bundles are all integral and $v_j(G) \geq \MMS_j$, by the pigeonhole principle, at least one of the three bundles is valued at least~$\MMS_j / 3$.

If there exists an agent~$j \in [3] \setminus i$ and two distinct bundles $G^*_a, G^*_b$ (out of the three bundles) such that $v_j(G^*_a), v_j(G^*_b) \geq \MMS_j / 3$, then for the last agent~$k$, we could allocate her most preferred bundle (which is valued at least~$\MMS_k / 3$), and still left with at least one bundle valued at least~$\MMS_j / 3$ for agent~$j$.
Allocating this bundle to agent~$j$ and the last bundle to agent~$i$ gives a $1/3$-MMS allocation.

If both agents~$j, k \in [3] \setminus i$ value only one bundle at least $\MMS_j / 3, \MMS_k / 3$, respectively, then there must exist one bundle, say $G^*_a$, such that $v_j(G^*_a) < \MMS_j / 3$ and $v_k(G^*_a) < \MMS_k / 3$.
By \Cref{lemma:reduced instance MMS}, it follows that under the reduced instance $I' = \langle \{j, k\}, G \setminus G^*_a \rangle$, we have that $\MMS_j(2, G \setminus G_a^*) \geq \frac{2 \cdot \MMS_j}{3}$ and $\MMS_k(2, G \setminus G_a^*) \geq \frac{2 \cdot \MMS_k}{3}$.
This implies that a $1/2$-MMS allocation in the reduced instance~$I'$ (guaranteed by \Cref{thm:2-agent:1/2-MMS-one-breakpoint}) is valued at least $\MMS_j / 3$ and $\MMS_k / 3$ for agents~$j$ and~$k$, respectively.
Therefore, allocating goods~$G \setminus G^*_a$ to agents~$j$ and~$k$ via \Cref{alg:2-agent-1/2MMS} together with allocating~$G^*_a$ to~$i$ gives a $1/3$-MMS allocation.

\paragraph{Case~2: $v_i(G_i) < \MMS_i$ for all~$i \in \agents$.}
By \Cref{lemma:small goods partition generalisation}, we have $v_i(G'_i) = v_i(G \setminus G_i) \geq 2/3 \cdot \MMS_i$.
For each distinct~$i, j \in \agents$, let $G'_{\{i, j\}} \coloneqq \{g \in G \colon c_{i, g} \leq 0.5 \land c_{j, g} \leq 0.5\}$.
Also, for notational convenience, in the following context, for any integral bundle~$\widetilde{G}$, denote by~$0.5 \widetilde{G} \coloneqq \{0.5 \cdot g \colon g \in \widetilde{G}\}$ the set in which exactly half of each good~$g \in \widetilde{G}$ is included.

We first consider the case where there exists~$G'_{\{i, j\}}$ such that $v_i(G'_{\{i, j\}}) \geq \MMS_i / 3$ and $v_j(G'_{\{i, j\}}) \geq \MMS_j / 3$.
Recall that all goods are valued less than~$\MMS_a / 3$ for all agents~$a \in \agents$.
Iteratively moving goods from~$G'_{\{i, j\}}$ would form a bundle $G^* \subseteq G'_{\{i, j\}}$ such that $v_i(G^*) \geq \MMS_i / 3$ and $v_j(G^*) \geq \MMS_j / 3$; also, assume without loss of generality $v_i(G^*) < \frac{2}{3} \cdot \MMS_i$.
Let~$k$ be the last agent:
\begin{itemize}
\item If $v_k(G^* \cap G'_k) \geq \MMS_k / 3$, then for agents~$j$ and~$k$, it follows that $v_j(0.5 G^*) \geq \MMS_j / 3$ and $v_k(0.5G^*) \geq \MMS_k / 3$.
Further, for agent~$i$, since $v_i(G^*) < \frac{2}{3} \cdot \MMS_i$, we have $v_i(G \setminus G^*) \geq \MMS_i / 3$.
Allocating half of all goods in~$G^*$ to agents~$j$ and~$k$ and $G \setminus G^*$ to agent~$i$ gives a $1/3$-MMS allocation.

\item Otherwise, $v_k(G^* \cap G'_k) < \MMS_k / 3$.
Then, $v_k(G \setminus G^*) \geq v_k(G'_k) - v_k(G^* \cap G'_k) > 2/3 \cdot \MMS_k - \MMS_k / 3 = \MMS_k / 3$.
Allocating half of all goods in~$G^*$ to agents~$i$ and~$j$ and $G \setminus G^*$ to agent~$k$ gives a $1/3$-MMS allocation.
\end{itemize}

We now consider the case where all $G'_{\{i, j\}}$'s are valued less than $\MMS_a / 3$ for some agent $a\in \{i, j\}$.
We pick an arbitrary~$G'_{\{i, k\}}$ such that $v_k(G'_{\{i, k\}}) < \MMS_k / 3$.
We iteratively move goods from~$G'_i$ to a bundle~$G^*_i$ until $v_i(G^*_i) \geq \MMS_i / 3$ and allocate goods~$0.5G^*_i$ to agents~$i$ and~$j$.
Then we iteratively add goods from~$G_j'$ to a bundle~$G^*_j$ until $v_j(0.5G^*_i \cup G^*_j) \geq \MMS_j / 3$ and allocate~$0.5G^*_j$ to agents~$j$ and~$k$.
Observe that since $v_i(G'_i) \geq 2/3 \cdot \MMS_i$ and $v_j(G'_j) \geq 2/3 \cdot \MMS_j$, $G^*_i$ and $G^*_j$ must exists and after the above allocations, agents $i$, $j$ must receive bundles $v_i(0.5G^*_i) \geq \MMS_i / 3$ and $v_j(0.5G^*_i \cup 0.5G^*_j) \geq \MMS_j / 3$ respectively.
It remains to show  that agent~$k$ values the remaining goods at at least~$\MMS_k / 3$.
With a slight abuse of notation, we have
\begin{align*}
v_k(G \setminus (G^*_i \cup 0.5G^*_j)) \geq v_k(G_k' \setminus (G^*_i \cup 0.5G^*_j)) = v_k((G_k' \setminus 0.5G^*_j) \setminus G^*_i) & = v_k(G_k' \setminus G^*_i) \\
& \geq \frac{2}{3} \cdot \MMS_k - \MMS_k / 3 \\
& = \MMS_k / 3,
\end{align*}
where the second last equality is due to all goods in~$G'_k$ having thresholds at most~$1/2$, and the last inequality being due to $v_k(G_k') \geq \frac{2}{3} \cdot \MMS_k$ and that $v_k(G_k' \cap G_i^*) \leq v_k(G_k' \cap G_i') = v_k(G'_{\{i, k\}}) < \MMS_k / 3$.

We conclude that \crefrange{alg:3-agent-1/3-MMS:small-goods-begins}{alg:3-agent-1/3-MMS:small-goods-ends} of our algorithm returns a $1/3$-MMS allocation in all cases, completing the proof.
\end{proof}

The proof of \Cref{thm:3-agent:1/3-MMS-one-breakpoint} follows from \Cref{lem:3-agent:1/3-MMS-one-breakpoint:large-goods,lem:3-agent:1/3-MMS-one-breakpoint:small-goods}.

\section{EF and Pareto Optimal Allocations}
\label{sec:EF}

In this section, we explore the widely recognized fairness concept of \emph{envy-freeness} in conjunction with the efficiency notion of Pareto optimality.
To set the stage, we first define the envy-freeness and Pareto optimality, and a combined notion, envy-freeness constrained Pareto optimality.

\begin{definition}[EF]
    An allocation~$\alloc$ is said to be \emph{envy-free (EF)} if, for every pair of agents $i, j \in \agents$, $v_i(A_i) \geq v_i(A_j)$.
\end{definition}

\begin{definition}[PO]
    Given an allocation~$\alloc$, another allocation~$\alloc' = (A'_i)_{i \in N}$ \emph{Pareto dominates}~$\alloc$ if $v_i(A'_i) \geq v_i(A_i)$ for all~$i \in N$ and $v_j(A'_j) > v_j(A_j)$ for some~$j \in N$.
    An allocation is said to be \emph{Pareto optimal (PO)} if no other allocation Pareto dominates it.
\end{definition}

\begin{definition}[EF-constrained-PO]
    An allocation~$\alloc$ is said to be \emph{envy-free constrained Pareto optimal (EF-constrained-PO)} if allocation~$\alloc$ is envy-free and no other envy-free allocation~$\alloc'$ exists such that $v_i(A'_i) \geq v_i(A_i)$ for every agent~$i \in N$ and $v_i(A'_j) > v_i(A_j)$ for some~$j \in N$.
\end{definition}

A PO allocation is guaranteed to exist for divisible goods, regardless of the utility function.\footnote{This can be seen from the fact that, for instance, an allocation which maximizes utilitarian welfare satisfies PO.}
While an EF allocation of multiple divisible goods also always exists in our context since each item can be equally divided among agents, such an allocation is often not PO.
This raises the issue of achieving both EF and PO simultaneously, and for what instances does an EF and PO allocation always exists.
Under the indivisible goods setting, it is well established that determining whether an EF allocation exists for multiple goods is NP-hard, and the problem of finding an allocation that satisfies both EF and PO is even more computationally complex~\citep{deKeijzerBoKl09}.
In our model, we show even when there are $3$~goods, determining whether an EF and PO allocation exists is NP-hard.

\begin{theorem}
    \label{thm:EF+PO-NP-hard}
    The existence of an EF and PO allocation is NP-hard for $m \geq 3$ and one-breakpoint piecewise-constant valuations.
\end{theorem}

\begin{proof}
Consider the following decision problem.

\pbDef{\textsc{3-EF-PO Existence}}
{An instance $I=(\agents,\divGoods)$ where $|\divGoods|\geq 3$}
{Does there exist an envy-free and Pareto optimal allocation in $I$?}

We will prove that \textsc{3-EF-PO Existence} is in NP-hard by reducing from the \textsc{Partition} problem.

\pbDef{\textsc{Partition}}
{A set of positive integers $S$}
{Does there exist a subset $S' \subseteq S$ such that $\sum_{s \in S'} s = \sum_{s \in S \setminus S'} s$?}

Given any instance of the \textsc{Partition} problem, we construct an instance of the \textsc{3-EF-PO Existence} problem with $|S| + 1$ agents, $\agents = [|S| + 1]$, and 3 goods, $\divGoods = \{g_1, g_2, g_3\}$. Each of the first $|S|$ agents has a piecewise-constant utility function:
$$v_i(g, p) = \mathbb{I}_{p \geq 2s_i / \sum_{s \in S} s} \text{ for both } g_1 \text{ and } g_2,$$
and
$$v_i(g_3, p) = \epsilon \mathbb{I}_{p \geq 1}.$$
The last agent, $|S| + 1$, has utility $v_{|S| + 1}(g_3, p) = \mathbb{I}_{p \geq 1}$ for $g_3$, and zero utility for $g_1$ and $g_2$.

        Suppose there exists an EF and PO allocation to the \textsc{3-EF-PO Existence} instance. Consider the allocation of $g_3$. Since the allocation must be PO, and all agents derive utility only if the entire good $g_3$ is allocated, one agent must receive the entirety of $g_3$. Given that the allocation is also EF, and agent $|S| + 1$ derives no utility from $g_1$ and $g_2$, it must be the case that agent $|S| + 1$ receives $g_3$; otherwise, agent $|S| + 1$ would envy the agent who received $g_3$.

        Now, since agent $|S| + 1$ receives $g_3$, every other agent $i \in [|S|]$ must receive a non-zero utility from the allocation of $g_1$ or $g_2$, or else they would envy agent $|S| + 1$. Additionally, since the total proportion required to satisfy all agents’ utility is 2, no agent can receive utility from both goods; otherwise, some agents would be left with no utility. Therefore, each agent must receive their required proportion from exactly one of the goods, either $g_1$ or $g_2$. By observing which agents receive $g_1$ and which receive $g_2$, we can solve the \textsc{Partition} problem.

        Conversely, if there is a solution to the \textsc{Partition} problem, we can allocate the appropriate proportion of $g_1$ to all agents in $S'$, $g_2$ to all agents in $S \setminus S'$, and the entirety of $g_3$ to agent $|S| + 1$. This results in an EF and PO allocation, thereby solving the \textsc{3-EF-PO Existence} instance.

        According to~\citep{GareyJo79}, \textsc{Partition} problem is NP-complete, the \textsc{3-EF-PO Existence} problem is also NP-hard.
\end{proof}

In the reminder of this section, we examine the case with a single good, which has been explored in related but different contexts~\citep{BuermannGeRa20,FeigeSaTa21}.

\subsection{Single Good Allocation}
In this subsection, we focus on the allocation of a single good, a scenario in which finding an EF and PO allocation is computationally tractable. We present an algorithm that efficiently computes an EF-constrained PO allocation for a single good, which can then be used to determine whether an EF and PO allocation exists. For ease of notation and since there is only $m = 1$ good concerned, we denote $v_i(p) = v_i(g, p)$ as a function of allocated proportion in this section only. Similarly, we denote $A_i = A_{i, g}$ as the proportion allocated to agent $i$ under $A$.

We now present the algorithm as follows.

Let $\underline{s_i}(b)$ represent the minimum proportion $s$ such that $v_i(s) = v_i(b)$; let $B = \{b_1, \ldots, b_k\}$ denotes all breakpoints across all agents.
The proposed algorithm operates in three stages to ensure EF-constrained PO in the allocation:
\begin{enumerate}
    \item Identify the largest breakpoint $b_{j}$ such that $\sum_{i \in N} \underline{s_i}(b_j) \leq 1$. Allocate to each agent $i$ a portion of the good such that $A_i = \underline{s_i}(b_j)$.
    \item Let $\epsilon$ denote a small remainder of the good, and define $L$ as the set of agents such that, for all $i \in L$, $\underline{s_i}(b_{j} + \epsilon) > b_j$. Adjust the allocation for each agent $i \in L$ as:
          $$ A_i = \min\left(b_{j+1} - \epsilon, b_{j} + \frac{1 - \sum_{k \in N} \underline{s_k}(b_j)}{|L|}\right),$$
          while keeping the allocation for other agents unchanged.
    \item Allocate the remaining portion of the good in such a way that no agent gets $b_{j+1}$ proportion of the good.
\end{enumerate}
Since there are polynomially many breakpoints, the algorithm runs in polynomial time. We now show that the algorithm returns an EF-constrained PO allocation.
\begin{algorithm}[t]
    \caption{EF-constrained PO allocation \label{alg:POEFAlloc}}
    \KwIn{An instance $I = (\agents, \divGoods)$.}
    \KwOut{An EF-constrained PO allocation $A$}
    \BlankLine
    \tcp{Stage 1}
    $j \gets 0$\;
    \While{$\sum_{i \in N} \underline{s_i}(b_{j+1}) \leq 1$}{
        $j = j + 1$\;
    }
    \ForEach{$i \in \agents$}{
        $A_i \gets \underline{s_i}(b_j)$\;
    }
    \tcp{Stage 2}
    $L \gets \varnothing$\;
    \ForEach{$i \in \agents$}{
        \If{$\underline{s_i}(b_{j} + \epsilon) > b_j$}{
            $L \gets L \cup i$\;
        }
    }
    \ForEach{$i \in L$}{
        $A_i \gets \min\left(b_{j+1} - \epsilon, b_j + \frac{1 - \sum_{i \in N} \underline{s_i}(b_j)}{|L|}\right)$\;
    }
    \tcp{Stage 3}
    \While{$\sum_{i \in \agents} A_i < 1$}{
        find $i$ such that $A_i < b_{j+1} - \epsilon$\;
        $A_i \gets \min(b_{j+1} - \epsilon, 1 + A_i - \sum_{j \in \agents} A_j)$\;
    }
    \Return $A$\;
\end{algorithm}

\begin{lemma}
    \label{lem:EFAlloc}
    The allocation found by Algorithm~\ref{alg:POEFAlloc} is EF.
\end{lemma}

\begin{proof}
        For each of the three stages in the algorithm, we show that the partial allocation remains EF at the end of the stage.
        \begin{enumerate}
            \item At the end of stage 1, each agent receives a utility equivalent to receiving a proportion $b_j$ of the good. Since no agent receives more than $b_j$, no agent envies another, and thus, the allocation is EF.
            \item In stage 2, since we define $L$ as the set of agents that derive more utility by receiving $\epsilon$ more proportion, all agents in $L$ must have received exactly $b_j$ proportion in stage 1. Therefore, all agents in $L$ receive the same additional proportion in stage 2 and do not envy one another. Further, since no agents receive at least $b_{j+1}$ and there are no breakpoints between $b_j$ and $b_{j+1}$, no agents outside of $L$ can increase their utility by receiving up to $b_{j+1} - \epsilon$. The allocation thus remains EF.
            \item In stage 3, additional goods are allocated only when all agents in $L$ received $b_{j+1} - \epsilon$ in stage 2. Since the remaining proportion of the goods are allocated such that no agent receives $b_{j+1}$ proportion, no agents in $L$ will envy any agents outside of $L$. Further, since there are no breakpoints between $b_j$ and $b_{j+1}$ and all agents receives proportion strictly less than $b_{j+1}$, no agent outside of $L$ will envy any other agent. Thus, the allocation remains EF.
        \end{enumerate}
        Since the allocation remains EF at the end of each stage, the allocation found by Algorithm~\ref{alg:POEFAlloc} is EF.
\end{proof}

\begin{lemma}
    \label{lem:POAlloc}
    No EF allocation Pareto dominates the allocation found by Algorithm~\ref{alg:POEFAlloc}.
\end{lemma}

\begin{proof}
        We have already shown that the allocation is EF at the end of stage 2. Now, we show that no agent can achieve the same utility while receiving a smaller proportion of the good.

        For each agent $i \in L$, since there are no breakpoints between $b_j$ and $A_i$, if $\underline{s_i}(b_j + \epsilon) > b_j$ holds, then $\underline{s_i}(p + \epsilon) > p$ must hold for all $p \in [b_j, b_{j+1})$. Therefore, no agent in $L$ can maintain the same utility while receiving less. Similarly, for each agent $i \notin L$, since $A_i = \underline{s_i}(b_j)$, no agent can maintain the same utility with a smaller allocation. As stage 3 does not change the agents' utilities, any allocation that Pareto dominates the one from stage 2 must allocate at least as much to each agent as the stage 2 allocation.

        In stage 3, if there are remaining goods, any allocation that increases some agents' utility will cause at least one agent to receive $b_{j+1}$. If there is remaining good, all agents in $L$ must receive $b_{j+1} - \epsilon$ (otherwise, they would have received more in stage 2). Once agents in $L$ receive additional goods, they will obtain at least $b_{j+1}$. Since there are no breakpoints between $b_j$ and $b_{j+1}$, agents not in $L$ cannot increase their utility by receiving up to $b_{j+1} - \epsilon$. Thus, any allocation that increases utility for some agents will result in at least one agent receiving $b_{j+1}$.

        Finally, since $\sum_{i \in N} \underline{s_i}(b_{j+1}) > 1$, any allocation with at least one agent receiving $b_{j+1}$ or more must leave some agent $i$ receiving strictly less than $\underline{s_i}(b_{j+1})$. Such an agent $i$ must envy the one receiving $b_{j+1}$. Therefore, no EF allocation can have any agents receiving at least $b_{j+1}$, and thus the allocation found by Algorithm~\ref{alg:POEFAlloc} is PO among all EF allocations.
\end{proof}

Combining \Cref{lem:EFAlloc} and \Cref{lem:POAlloc}, we have shown that the allocation found by Algorithm~\ref{alg:POEFAlloc} is EF-constrained PO, and thus the following theorem follows.

\begin{theorem}
    \label{thm:EF-EFPO-allocation}
    An EF-constrained PO allocation always exist and can be found in polynomial time for a single divisible good with piecewise-linear utility functions.
\end{theorem}

Further, one can easily show that any EF and PO allocation must itself be EF-constrained PO and PO. However, it turns out a stronger result holds for $m = 1$ good: an EF and PO allocation exists if and only if the EF-constrained PO allocation found by Algorithm~\ref{alg:POEFAlloc} is also PO. This is proved in the following theorem.

\begin{theorem}
    \label{thm:EF-PO-existence}
    The existence of an EF and PO allocation can be checked in polynomial time for a single divisible good with piecewise-linear utility functions.
\end{theorem}

\begin{proof}
        We use Algorithm~\ref{alg:POEFAlloc} to compute an EF-constrained PO allocation and show that an EF and PO allocation exists if and only if the returned allocation is PO.

        First, suppose that the returned allocation is PO, but since it already is EF-constrained PO (and hence EF), it is trivial that the allocation is EF and PO.

        Conversely, suppose that the returned allocation is not PO, then there must exist another allocation $A'$ that Pareto dominates it. In particular, by the construction of Algorithm~\ref{alg:POEFAlloc}, this implies that all agents in $L$ must receive $b_{j+1} - \epsilon$ in stage 2 (or otherwise, if they receive $b_{j} + \frac{1 - \sum_{k \in N} \underline{s_k}(b_j)}{|L|}$, then any $\epsilon$ proportion taken away from any agents must reduce that agent's utility, and so the returned allocation must be PO). It then follows that under $A'$, all agents in $L$ must receive at least $b_{j+1} - \epsilon$; all agents outside of $L$ must receive at least $b_{j}$; and some agent must receive at least $b_{j+1}$. However, the last paragraph in the proof of \Cref{lem:POAlloc} implies that no EF allocation can have any agents receiving $b_{j+1}$ or more. Therefore, for any agent $i$, the maximum utility in any EF allocation is $v_i(b_{j+1} - \epsilon)$ if $i\in L$ and $v_i(b_{j})$ if $i\notin L$. This implies that any EF allocation must be Pareto dominated by $A'$. Thus, no EF and PO allocation exists.

        Finally, checking whether an allocation is PO for a single divisible good with piecewise-linear utility functions can be done in polynomial time: for each agent, we compute the maximum portion of the good that can be removed while maintaining the same utility, denoted as $w_i = A_i - \underline{s_i}(A_i)$. We remove $w_i$ from each agent and check whether there exists an agent $j$ such that $v_j(A(j) + \sum_{i \in \agents} w_i) > v_j(A(j))$. If such an agent exists, the allocation is not PO. Otherwise, the allocation is PO.
\end{proof}

\section{Discussion}

In this paper, we have studied a fair division model where a set of homogeneous divisible goods are allocated among agents who have non-linear valuations in the sense that each agent's value depends only on the amount of each good she receive.
We focus on the fair-share-based notion of MMS and the comparison-based fairness notion of envy-freeness.
In more detail, we give tight or asymptotically tight approximation to MMS in various cases.
For envy-freeness, we explore computational complexity of checking the existence of an envy-free and PO allocation.
Since one-breakpoint piecewise-constant valuations capture a wide range of real-world applications and can be easily elicited from agents in practice, it would be intriguing to fully understand whether we can guarantee $\frac{1}{n}$-MMS for $n$ agents---we conjecture that the answer is affirmative, and the computational complexity of checking the existence of an EF and PO allocation for exactly two goods.

In future research, since complementarity and substitutability are common in practice, it would be interesting to consider compact but more expressive valuations that relax our assumption of values across different goods being additive.
As one possibility, we could let each agent has a utility function $v((p_j)_{j \in [m]})$ that specifies the value for getting fraction $p_1$ of item~$g_1$, $p_2$ of item~$g_2$, and so on.

\section*{Acknowledgements}

This work was partially supported by the NSF-CSIRO grant on ``Fair Sequential Collective Decision-Making'' (Grant No.\ RG230833),
by the ARC Laureate Project FL200100204 on ``Trustworthy AI'', and
by JST ERATO Grant Number JPMJER2301.

\bibliographystyle{plainnat}
\bibliography{bibliography}
\end{document}